\documentclass{lmcs} 
\pdfoutput=1

\usepackage{lastpage}
\lmcsdoi{16}{1}{19}
\lmcsheading{}{\pageref{LastPage}}{}{}%
{Feb.~25,~2019}{Feb.~17,~2020}{}


\usepackage{amsfonts}
\usepackage{amssymb}
\usepackage{graphicx}
\usepackage{bussproofs}
\usepackage{tensor,color}
\usepackage{url}
\usepackage{lscape}
\usepackage{esvect}

\usepackage{todonotes}

\newcommand{\figureboxed}[1]{%
  \hrule
  #1
  \hrule
}

\newdimen\proofrulebreadth \proofrulebreadth=.05em
\newdimen\proofdotseparation \proofdotseparation=1.25ex
\newdimen\proofrulebaseline \proofrulebaseline=2ex
\newdimen\proofrulebaseline \proofrulebaseline=1.7ex
\newcount\proofdotnumber \proofdotnumber=3
\let\then\relax
\def\hfi{\hskip0pt plus.0001fil}
\mathchardef\squigto="3A3B
\newif\ifinsideprooftree\insideprooftreefalse
\newif\ifonleftofproofrule\onleftofproofrulefalse
\newif\ifproofdots\proofdotsfalse
\newif\ifdoubleproof\doubleprooffalse
\let\wereinproofbit\relax
\newdimen\shortenproofleft
\newdimen\shortenproofright
\newdimen\proofbelowshift
\newbox\proofabove
\newbox\proofbelow
\newbox\proofrulename
\def\shiftproofbelow{\let\next\relax\afterassignment\setshiftproofbelow\dimen0 }
\def\shiftproofbelowneg{\def\next{\multiply\dimen0 by-1 }%
\afterassignment\setshiftproofbelow\dimen0 }
\def\setshiftproofbelow{\next\proofbelowshift=\dimen0 }
\def\setproofrulebreadth{\proofrulebreadth}

\def\prooftree{
\ifnum  \lastpenalty=1
\then   \unpenalty
\else   \onleftofproofrulefalse
\fi
\ifonleftofproofrule
\else   \ifinsideprooftree
        \then   \hskip.5em plus1fil
        \fi
\fi
\bgroup
\setbox\proofbelow=\hbox{}\setbox\proofrulename=\hbox{}%
\let\justifies\proofover\let\leadsto\proofoverdots\let\Justifies\proofoverdbl
\let\using\proofusing\let\[\prooftree
\ifinsideprooftree\let\]\endprooftree\fi
\proofdotsfalse\doubleprooffalse
\let\thickness\setproofrulebreadth
\let\shiftright\shiftproofbelow \let\shift\shiftproofbelow
\let\shiftleft\shiftproofbelowneg
\let\ifwasinsideprooftree\ifinsideprooftree
\insideprooftreetrue
\setbox\proofabove=\hbox\bgroup$\displaystyle 
\let\wereinproofbit\prooftree
\shortenproofleft=0pt \shortenproofright=0pt \proofbelowshift=0pt
\onleftofproofruletrue\penalty1
}

\def\eproofbit{
\ifx    \wereinproofbit\prooftree
\then   \ifcase \lastpenalty
        \then   \shortenproofright=0pt  
        \or     \unpenalty\hfil         
        \or     \unpenalty\unskip       
        \else   \shortenproofright=0pt  
        \fi
\fi
\global\dimen0=\shortenproofleft
\global\dimen1=\shortenproofright
\global\dimen2=\proofrulebreadth
\global\dimen3=\proofbelowshift
\global\dimen4=\proofdotseparation
\global\count255=\proofdotnumber
$\egroup  
\shortenproofleft=\dimen0
\shortenproofright=\dimen1
\proofrulebreadth=\dimen2
\proofbelowshift=\dimen3
\proofdotseparation=\dimen4
\proofdotnumber=\count255
}

\def\proofover{
\eproofbit 
\setbox\proofbelow=\hbox\bgroup 
\let\wereinproofbit\proofover
$\displaystyle
}%
\def\proofoverdbl{
\eproofbit 
\doubleprooftrue
\setbox\proofbelow=\hbox\bgroup 
\let\wereinproofbit\proofoverdbl
$\displaystyle
}%
\def\proofoverdots{
\eproofbit 
\proofdotstrue
\setbox\proofbelow=\hbox\bgroup 
\let\wereinproofbit\proofoverdots
$\displaystyle
}%
\def\proofusing{
\eproofbit 
\setbox\proofrulename=\hbox\bgroup 
\let\wereinproofbit\proofusing
\kern0.3em$
}

\def\endprooftree{
\eproofbit 
  \dimen5 =0pt
\dimen0=\wd\proofabove \advance\dimen0-\shortenproofleft
\advance\dimen0-\shortenproofright
\dimen1=.5\dimen0 \advance\dimen1-.5\wd\proofbelow
\dimen4=\dimen1
\advance\dimen1\proofbelowshift \advance\dimen4-\proofbelowshift
\ifdim  \dimen1<0pt
\then   \advance\shortenproofleft\dimen1
        \advance\dimen0-\dimen1
        \dimen1=0pt
        \ifdim  \shortenproofleft<0pt
        \then   \setbox\proofabove=\hbox{%
                        \kern-\shortenproofleft\unhbox\proofabove}%
                \shortenproofleft=0pt
        \fi
\fi
\ifdim  \dimen4<0pt
\then   \advance\shortenproofright\dimen4
        \advance\dimen0-\dimen4
        \dimen4=0pt
\fi
\ifdim  \shortenproofright<\wd\proofrulename
\then   \shortenproofright=\wd\proofrulename
\fi
\dimen2=\shortenproofleft \advance\dimen2 by\dimen1
\dimen3=\shortenproofright\advance\dimen3 by\dimen4
\ifproofdots
\then
        \dimen6=\shortenproofleft \advance\dimen6 .5\dimen0
        \setbox1=\vbox to\proofdotseparation{\vss\hbox{$\cdot$}\vss}%
        \setbox0=\hbox{%
                \advance\dimen6-.5\wd1
                \kern\dimen6
                $\vcenter to\proofdotnumber\proofdotseparation
                        {\leaders\box1\vfill}$%
                \unhbox\proofrulename}%
\else   \dimen6=\fontdimen22\the\textfont2 
        \dimen7=\dimen6
        \advance\dimen6by.5\proofrulebreadth
        \advance\dimen7by-.5\proofrulebreadth
        \setbox0=\hbox{%
                \kern\shortenproofleft
                \ifdoubleproof
                \then   \hbox to\dimen0{%
                        $\mathsurround0pt\mathord=\mkern-6mu%
                        \cleaders\hbox{$\mkern-2mu=\mkern-2mu$}\hfill
                        \mkern-6mu\mathord=$}%
                \else   \vrule height\dimen6 depth-\dimen7 width\dimen0
                \fi
                \unhbox\proofrulename}%
        \ht0=\dimen6 \dp0=-\dimen7
\fi
\let\doll\relax
\ifwasinsideprooftree
\then   \let\VBOX\vbox
\else   \ifmmode\else$\let\doll=$\fi
        \let\VBOX\vcenter
\fi
\VBOX   {\baselineskip\proofrulebaseline \lineskip.2ex
        \expandafter\lineskiplimit\ifproofdots0ex\else-0.6ex\fi
        \hbox   spread\dimen5   {\hfi\unhbox\proofabove\hfi}%
        \hbox{\box0}%
        \hbox   {\kern\dimen2 \box\proofbelow}}\doll%
\global\dimen2=\dimen2
\global\dimen3=\dimen3
\egroup 
\ifonleftofproofrule
\then   \shortenproofleft=\dimen2
\fi
\shortenproofright=\dimen3
\onleftofproofrulefalse
\ifinsideprooftree
\then   \hskip.5em plus 1fil \penalty2
\fi
}
\newcommand{\cent}{\vdash}

\newcommand{\theory}[1]{\ensuremath{\mathsf{#1}}}

\newcommand\aleq{\mathrel{{\approx}_{\scriptstyle {\alpha}}}}

\newcommand{\caleq}{\mathrel{\stackrel{\fixp}{\approx}_{\scriptstyle {\alpha, \theory{C}}}}}

\newcommand{\ealeq}[1]{\mathrel{\stackrel{\fixp}{\approx}_{\scriptstyle {\alpha, \theory{#1}}}}}

\newcommand\act{\cdot}

\newcommand{\rulefont}[1]{\ensuremath{(\mathbf{#1})}}

\newcommand{\tf}[1]{\mathsf{#1}}

\newcommand{\Id}{\mathit{Id}}

\newcommand{\fixp}{\ensuremath\curlywedge}

\newcommand{\efixp}[1]{\ensuremath\curlywedge_{\theory{#1}}}

\newcommand{\new}{\reflectbox{$\mathsf N$}}

\newcommand{\upsvar}[1]{\Upsilon|_{#1}}

\newcommand{\perm}[1]{\texttt{perm}(#1)}
\newcommand{\supp}[1]{\texttt{supp}(#1)}
\newcommand{\suppt}[2]{\texttt{supp}_{#1}(#2)}

\newcommand{\swap}[2]{(#1 \ #2)}
\newcommand{\var}[1]{\texttt{Var}(#1)}
\newcommand{\pair}[2]{\langle #1, #2\rangle}
\newcommand{\nftriple}[1]{\langle \mathcal{#1}\rangle_\texttt{nf}}
\newcommand{\probl}{\texttt{Pr}}
\newcommand{\permset}{\texttt{Perm}(\mathbb{A})}
\newcommand{\diffs}[2]{\texttt{ds}(#1, #2)}
\newcommand{\dom}[1]{\texttt{dom}(#1)}
\newcommand{\fn}[1]{{\tt fn}(#1)}

\newcommand{\appAE}{\approx_{\{\alpha, \theory{E}\}}}
\newcommand{\appAC}{\approx_{\{\alpha, \theory{C}\}}}
\newcommand{\appAA}{\approx_{\{\alpha, \theory{A}\}}}
\newcommand{\appAAC}{\approx_{\{\alpha, \theory{AC}\}}}
\newcommand{\appAll}{\approx_{\{\theory{A},\theory{C},\theory{AC}\}}}
\newcommand{\faleq}{\stackrel{\fixp}{\approx}_{\alpha}}
\newcommand{\ucaleq}{\stackrel{\fixp ?}{\approx}_\theory{C}}

\newcommand\sol[1]{\langle #1\rangle_{\mathtt{\it{sol}}}}

\theoremstyle{plain}

\newtheorem{proposition}[thm]{Proposition}
\newtheorem{corollary}[thm]{Corollary}

\theoremstyle{plain} 

\begin{document}

\title[Nominal Syntax and Permutation Fixed Points]{On Nominal Syntax and Permutation Fixed Points}

\author[M.~Ayala-Rinc\'on]{Mauricio Ayala-Rinc\'on}
\address{Departments Mathematics and Computer Science, Universidade de Bras\'ilia, Bras\'ilia, Brazil}
\email{ayala@unb.br}
\thanks{Work supported by FAPDF grant 193001369/2016. M. Ayala-Rinc\'on partially funded by CNPq research grant number 307672/2017-4.}

\author[M.~Fern\'andez]{Maribel Fern\'andez}
\address{Department of Informatics, King's College London, London, UK}	
\email{maribel.fernandez@kcl.ac.uk}

\author[D.~Nantes]{Daniele Nantes-Sobrinho}	
\address{Department Mathematics, Universidade de Bras\'ilia, Bras\'ilia, Brazil}
\email{dnantes@unb.br}


\begin{abstract}
  \noindent We propose a new axiomatisation of the alpha-equivalence relation for nominal terms, based on a primitive notion of fixed-point constraint. We show that the standard
freshness relation between atoms and terms can be derived from the more primitive notion of permutation fixed-point, and use this result to prove the correctness of the new $\alpha$-equivalence axiomatisation. 
This gives rise to a new notion of nominal unification, where solutions for unification problems are pairs of a fixed-point context and a substitution.
Although it may seem less natural than the standard notion of nominal unifier based on freshness constraints, the notion of unifier based on fixed-point constraints  behaves better when equational theories are considered: for example, nominal unification remains finitary in the presence of commutativity, whereas this is not the case 
when unifiers are expressed using freshness contexts. 
We provide a definition of $\alpha$-equivalence modulo equational theories that takes into account  $\theory{A}$, $\theory{C}$ and $\theory{AC}$ theories.
Based on this notion of equivalence, we show that $\theory{C}$-unification is finitary and we provide a sound and complete \theory{C}-unification algorithm, as a first step  towards the development of nominal unification modulo \theory{AC} and other equational theories with permutative properties. 

\end{abstract}

\maketitle

\section{Introduction}\label{S:one}

This paper presents a new approach for the definition of nominal languages, based on the use of permutation fixed points. More precisely, we give  a new axiomatisation of the $\alpha$-equivalence relation for nominal terms using permutation fixed-points, and revisit nominal unification in this setting. 

In nominal syntax~\cite{Urban2004},
\emph{atoms} are used to represent object-level variables
  and  \emph{atom permutations} are used  to implement renamings, following the nominal-sets approach advocated by Gabbay and Pitts~\cite{Gabbay2000,Gabbay2002a,pitts2013nominal}. Atoms can be abstracted over terms; the syntax $[a]s$ represents the abstraction of $a$ in $s$. To rename an abstracted atom $a$ to $b$, a \emph{swapping}
  permutation $\pi = (a\,b)$ is applied.
  Thus, the action of $\pi$ over  $[a]s$, written as $(a\,b)\act
  [a]s$,  produces the nominal term $[b]s'$, where $s'$ is the result of
  swapping  $a$ and $b$ in $s$. The $\alpha$-equivalence relation between nominal terms is specified using swappings together with a \emph{freshness relation} between atoms and terms, written $b\# s$,
  which roughly corresponds to $b$ not occurring free in $s$.

In this setting, checking $\alpha$-equivalence requires another
first-order specialised calculus to check freshness constraints. 
For instance, checking whether $[a]s
\approx_\alpha [b]t$ reduces to checking whether $s\approx_\alpha (b\,a)\act t$
and $a\# t$.   

The action of a permutation propagates
down the structure of nominal terms,  until a variable is reached: permutations suspend over variables.  Thus, $\pi\act s$ represents the action of a permutation
over a nominal term, but is not itself a nominal term unless $s$ is a
variable; for instance,  $\pi\act X$ is a \emph{suspension} (also called \emph{moderated variable}), which is a nominal term. 

The presence of moderated variables and atom-abstractions makes reasoning about equality
of nominal terms more involved than in standard first-order syntax. 
For example, $\pi\act X \approx_\alpha^? \rho\act X$  is only true 
when $X$ ranges over nominal terms, say $s$,  for
which all atoms in the difference set of $\pi$ and $\rho$ (i.e., the
set $\{a : \pi(a) \neq \rho(a)\}$) are fresh in $s$. 

If the domain of a permutation $\pi$ is fresh for $X$ then $\pi\act X \aleq Id\act X$.
Thus a set of freshness constraints (i.e., a freshness context) can be used to specify 
that a permutation will have no effect on the instances of $X$.
This is why in \emph{nominal unification}~\cite{Urban2004}, the solution for a problem is a pair consisting of a freshness context and a substitution.

The use of freshness contexts is natural when dealing with ``syntactic'' nominal unification, but in the presence of equational axioms (i.e.,  equational nominal unification) it is not straightforward. 
For example, in the case of \theory{C}-nominal unification (nominal unification modulo commutativity), to specify that a permutation has no effect on the instances of $X$ modulo \theory{C}, in other words, to specify that the permutation does not affect a given \theory{C}-equivalence class, we need something more than a freshness constraint (note that 
$\swap{a}{b}\act (a + b) = b+a =_\theory{C} a+b$, so the permutation $\swap{a}{b}$ fixes the term  $a + b$, despite the fact that $a$ and $b$ are not fresh).

In this paper, we propose to axiomatise $\alpha$-equivalence of nominal terms using permutation fixed-point constraints: we write $\pi \fixp t$ (read ``$\pi$ fixes $t$'') if $t$ is a fixed-point of $\pi$. We show how to derive fixed-point constraints $\pi \fixp t$ from primitive constraints of the form $\pi \fixp X$, and show the correctness of this approach by proving that the $\alpha$-equivalence relation generated in this way coincides with the one axiomatised via freshness constraints. 
We then show how   fixed-point  constraints can be used to specify $\alpha$-equivalence modulo equational theories containing $\theory{A}$, $\theory{C}$, and $\theory{AC}$ operators, and provide an algorithm to solve nominal unification  problems modulo $\theory{C}$, which outputs a finite set of most general solutions.

\paragraph{\bf Related Work}
Equational reasoning has been extensively explored since the early development of modern abstract algebra (see, e.g., the $E$-unification surveys by Siekmann~\cite{Siekmann90} and  Baader et al~\cite{BaaderSiekmann94} and \cite{Baader2001}).  For \theory{AC}-equality  checking, \theory{AC} matching and \theory{AC} unification, refined techniques have been applied. For instance, \theory{AC}-equality check and linear \theory{AC}-matching problems can be reduced 
to searching a perfect matching in a bipartite 
graph~\cite{BKN1987}, whereas \theory{AC} unification 
problems can be reduced to solving a system of Diophantine equations~\cite{Stickel81,fages1987}.

 Techniques to deal with $\alpha$-equivalence modulo the equational theories $\theory{A}$, $\theory{C}$ and $\theory{AC}$ were proposed in   \cite{Ayala-Rincon2016, ACFONantes_jv19, Ayala2017, Ayala-Rincon2018},  using the standard nominal approach via freshness constraints.
Solving nominal $\theory{C}$-unification problems requires to deal with fixed-point equations, for which there is no finitary representation of the set of solutions using only freshness constraints and substitutions~\cite{Ayala2017, Ayala-Rincon2018}. A combinatorial algorithm permits to find all the solutions of fixed-point equations~\cite{Ayala-Rincon2018}.

Fixed-point constraints arise also in other contexts: Schmidt-Schau\ss~ et al~\cite{Kutsia2016} show how nominal unification problems in a language with recursive let operators give rise to freshness constraints and nominal fixed-point equations. The approach to nominal unification via permutation fixed-points proposed in this paper could also be used to reason about equality in this language.


This paper is a revised and extended version of~\cite{DBLP:conf/rta/Ayala-RinconFN18}, where nominal terms with fixed-point permutation constraints were first presented. In this paper we prove the correctness of the approach, by showing it is equivalent to the standard presentation via freshness constraints. We also provide proofs of soundness and completeness of the unification algorithm.  A generalisation of the notion of  $\alpha$-equivalence to  take into account equational theories including $\theory{A}$, $\theory{C}$ and $\theory{AC}$ is also provided  (only $\theory{C}$ was considered in~\cite{DBLP:conf/rta/Ayala-RinconFN18}),  as well as a $\theory{C}$-unification algorithm.

\paragraph{\bf Overview}
The paper is organised as follows. Section~\ref{sec:preliminaries} provides the necessary background material on nominal syntax and semantics. Section~\ref{sec:constraints} introduces fixed-point constraints and $\alpha$-equivalence, and shows that these relations behave as expected. In particular, we show that there is a two-way translation between the freshness-based $\alpha$-equivalence relation and its permutation fixed-point counter-part, which confirms that the fixed-point approach is equivalent to the standard approach via freshness constraints. 
Section~\ref{sec:unif} presents a nominal unification algorithm specified by a set of simplification rules, and proves its soundness and completeness. In Section~\ref{sec:alpha-equivalence-modulo} we generalise the approach to take into account equational theories: we define a notion of permutation fixed-point $\alpha$-equivalence  modulo $\theory{A}$, $\theory{C}$ and $\theory{AC}$ theories, and develop a $\theory{C}$-unification algorithm. Finally, Section~\ref{sec:conclusions} concludes and discusses future work.

\section{Preliminaries}
\label{sec:preliminaries}
 Let $\mathbb{A}$ be a fixed and countably infinite set of elements $a,b,c,\ldots$, which will be called \emph{atoms} (atomic names). 
A permutation on  $\mathbb{A}$ is a bijection on  $\mathbb{A}$ with finite domain. 

Fix a countably infinite set $\mathcal{X}=\{X,Y,Z,\ldots\}$ of variables and a countable set $\mathcal{F}= \{ \tf{f}, \tf{g}, \ldots\}$ of function symbols.

\begin{defi}[Nominal grammar] Nominal terms are generated by the following grammar.
$$s, t :=  a \mid   [a]t \mid (t_1,\ldots, t_n) \mid  \tf{f}^\theory{E}\ t \mid  \pi \cdot X $$
where  $a$ is an {\em atom term}, $[a]t$ denotes the {\em abstraction}
of the atom $a$ over the term $t$, $(t_1,\ldots, t_n)$ is a tuple, function symbols are equipped with an equational theory $\theory{E}$, hence $\tf{f}^\theory{E}\ t$ denotes the {\em application of $\tf{f}^\theory{E}$ to $t$}
and $\pi \act X$ is a {\em moderated variable} or {\em suspension}, where $\pi$ is an atom permutation. We write $\tf{f}^\emptyset$, or simply $\tf{f}$, to emphasise that no equational theory is assumed for $\tf{f}$, that is, $\tf{f}$ is just a function symbol. 
\end{defi}

We follow the \textit{permutative convention}~\cite[Convention~2.3]{Gabbay2008} for atoms
throughout the paper, i.e., atoms $a, b, c$ range permutatively over $\mathbb{A}$ so that they are always pairwise different, unless stated otherwise.

Atom {\em permutations} are represented by finite lists of \emph{swappings}, which are
pairs of different atoms $\swap{a}{b}$; hence, a permutation $\pi$  is generated by the following grammar: 
$$\pi := Id \ | \ \swap{a}{b}\pi.$$
where $Id$ is the identity permutation,  usually omitted from the list of swappings representing a permutation $\pi$. 
Suspensions of the form $Id\act X$ will be represented just by $X$.  
We write $\pi^{-1}$ for the \textit{inverse} of $\pi$, 
and use $\circ$ to denote the composition of permutations. For example, if
\(\pi = \swap{a}{b}\swap{b}{c}\) then \(\pi(c) = a\) and \(c = \pi^{-1}(a)\). 

The \emph{difference set} of two permutations $\pi$, $\pi'$ is
$\diffs{\pi}{\pi'} = \{a \mid \pi(a) \neq \pi'(a)\}$.

We write $\var{t}$ for the set of variables occurring in $t$. Ground terms are terms
without variables, that is, $\var{t} = \emptyset$. A ground term may still contain atoms,
for example $a$ is a ground term and $X$ is not.

\begin{defi}[Permutation action] \label{def:perm} The action of a permutation $\pi$ on a term $t$ is
defined by induction on the number of swappings in $\pi$:

$ Id\act t = t$ and $(\swap{a}{b}\pi )\act t = \swap{a}{b}\act(\pi \act  t),$
where
\[
\begin{array}{@{\hspace{-1mm}}c@{\hspace{-.2mm}}c@{\hspace{-.2mm}}c}
\begin{array}{c}
\swap{a}{b}\act a = b \\
\swap{a}{b}\act  b = a \\ 
\swap{a}{b} \act  c = c
\end{array}
& 
\begin{array}{c}
\swap{a}{b}\act  (\pi \act  X) = (\swap{a}{b} \circ \pi)\act  X \\
\swap{a}{b}\act  \tf{f}\ t = \tf{f}\  \swap{a}{b}\act  t \\  
\swap{a}{b} \act  (t_1 ,\ldots, t_n ) = (\swap{a}{b}\act  t_1 ,\ldots, \swap{a}{b}\act t_n )
\end{array}
&
\begin{array}{c}
\swap{a}{b}\act  [a]t = [b]\swap{a}{b}\act t\\
\swap{a}{b}\act  [b]t = [a]\swap{a}{b}\act t\\
\swap{a}{b}\act  [c]t = [c]\swap{a}{b}\act t\\

\end{array}
\end{array}
\]
\end{defi}

\begin{defi}[Substitution] 
 {\em Substitutions} are generated by the grammar
$$\sigma ::= id \ | \ [X\mapsto s]\sigma.$$
Postfix notation is used for substitution application and $\circ$ for composition: $t(\sigma \circ \sigma') =(t\sigma)\sigma'$. Substitutions act on terms elementwise in the natural way:
$t \ id = t, \ 
t[X\mapsto s]\sigma$ $=$ $(t[X\mapsto s])\sigma$, 
 where
\[
\begin{array}{cc}
\begin{array}{c}
a[X\mapsto s] = a\\
(\tf{f}\ t)[X\mapsto s] = \tf{f} (t[X \mapsto s])\\
([a]t)[X\mapsto s] = [a](t[X\mapsto s])
\end{array}
&
\begin{array}{c}
(t_1 ,\ldots, t_n )[X\mapsto s] = (t_1 [X\mapsto s], \ldots, t_n [X\mapsto s])\\
(\pi \cdot X)[X\mapsto s] =\pi \cdot s\\
(\pi \cdot Y )[X \mapsto s] = \pi \cdot Y
\end{array}
\end{array}
\]
\end{defi}

The following well-known property of substitution and permutation justifies the notation $\pi\act s\sigma$ (without brackets), see ~\cite{Urban2004} for more details. 

\begin{lem}
\label{lem:subst-perm-commute}
Substitution and permutation commute:
$\pi\act (s\sigma) = (\pi\act s)\sigma$.
\end{lem}

\subsection{Alpha-equivalence via freshness constraints}
\label{sec:alpha-fresh}
In the standard nominal approach (see, e.g., ~\cite{pitts2003nominal,Urban2004,Fernandez2007}),  the $\alpha$-equivalence relation $s\aleq t$ is defined using a freshness relation between atoms and terms, written $a\# t$ -- read ``$a$ fresh for $t$'' -- which, intuitively,  corresponds to the idea of an atom not occurring free in a term. 
These relations are axiomatised using the rules in Figures~\ref{fig.rules.freshness} and \ref{fig.rules.equaf}, respectively.

\begin{figure*}[ht]
\small{
\figureboxed{
$$
\begin{array}{c@{\hspace{1.1cm}}c@{\hspace{1.1cm}}c}
\begin{prooftree}
\justifies \Delta\cent a\# b
\using \rulefont{\# a}
\end{prooftree}
&
\begin{prooftree}
\pi^{-1}(a) \# X \in \Delta
\justifies \Delta\cent a\#\pi \cdot X
\using \rulefont{\# var}
\end{prooftree}
&
\begin{prooftree}
\Delta\cent a\# t
\justifies 
\Delta\cent a\# \tf{f}\ t
\using \rulefont{\# \tf{f}} 
\end{prooftree}
\\[4ex]
\begin{prooftree}
\Delta\cent a\# t_1  \quad \ldots \quad \Delta\cent a\# t_n
\justifies
\Delta\cent a \#  (t_1,\ldots, t_n)
\using \rulefont{\# tuple}
\end{prooftree}
&
\begin{prooftree}
\justifies 
\Delta\cent a \#  [a]t
\using \rulefont{\# [a]}
\end{prooftree}
&
\begin{prooftree}
\Delta \cent a \# t
\justifies 
\Delta\cent a \#  [b]t
\using \rulefont{\# abs}
\end{prooftree}
\end{array}
$$}}
\caption{Rules for freshness}
\label{fig.rules.freshness}
\end{figure*}

We call $s \aleq t$ and $a\# t$ $\alpha$-equality and freshness constraints, respectively.
Note that to define $\aleq$ we use the  difference set of two permutations in rule $\rulefont{\aleq var}$; we denote by $\diffs{\pi}{\pi'}\# X$ the following set of freshness constraints:
$$\diffs{\pi}{\pi'}\# X = \{a \# X \ |\ a\in \diffs{\pi}{\pi'}\}.$$

\begin{figure*}[ht]
\small{
\figureboxed{$$
\begin{array}{c@{\hspace{1cm}}c}
\begin{prooftree}
\justifies \Delta \cent a \aleq a
\using \rulefont{\aleq a}
\end{prooftree}
&
\begin{prooftree}
\diffs{\pi}{\pi'}\# X \subseteq \Delta
\justifies 
\Delta\cent  \pi \cdot X \aleq \pi'\cdot X
\using \rulefont{\aleq var} 
\end{prooftree}
\\[4ex]
\begin{prooftree}
\Delta\cent t\aleq t'
\justifies 
\Delta\cent  \tf{f}\ t \aleq \tf{f}\ t'
\using \rulefont{\aleq\tf{f}} 
\end{prooftree}
&
\begin{prooftree}
\Delta\cent t_1\aleq t_1' \quad \ldots \quad\Delta \cent t_n\aleq t_n'
\justifies 
\Delta\cent  (t_1,\ldots, t_n) \aleq (t_1',\ldots, t_n')
\using \rulefont{\aleq tuple} 
\end{prooftree}
\\[4ex]
\begin{prooftree}
\Delta\cent  t\aleq t' 
\justifies 
\Delta\cent  [a]t \aleq [a]t'
\using \rulefont{\aleq [a]} 
\end{prooftree}
&
\begin{prooftree}
\Delta\cent\ s\aleq (a \ b )\act t \quad \Delta \cent  a \# t
\justifies
\Delta\cent [a] s \aleq [b] t
\using \rulefont{\aleq {\tf ab}} 
\end{prooftree}
\end{array}
$$}
}
\caption{Rules for $\alpha$-equality via freshness}
\label{fig.rules.equaf}
\end{figure*}

The symbols $\Delta$ and $\nabla$ denote \emph{freshness contexts}, which are sets of freshness constraints of the form $a\# X$. The domain of a freshness context $\Delta$, denoted by $\dom{\Delta}$,  consists of the atoms occurring in $\Delta$; $\Delta|_X$ consists of the restriction of $\Delta$ to the freshness constraints on variable $X$, that is, the set $\{a\# X \ | \ a\# X \in \Delta\}$. 


\subsection{Nominal sets and support}

Let  $S$ be a set equipped with an action of the group $\permset$ of finite permutations of $\mathbb{A}$. 

\begin{defi}
\label{def:support}
A set $A\subset \mathbb{A}$  is a \emph{support} for an element $x\in S$ if for all $\pi \in \permset$, the following holds
\begin{equation}
((\forall a \in A) \  \pi(a)=a) \Rightarrow \pi \cdot x = x
\end{equation}

A \emph{nominal set} is a set equipped with an action of the group $\permset$, that is, a $\permset$-set, all of whose elements have finite support. 
\end{defi}

As in~\cite{pitts2013nominal}, we denote by $\texttt{ supp}_S (x)$ the least finite support of $x$, that is, 
\begin{equation*}
\texttt{supp}_S(x):= \bigcap \{A \in \mathcal{P}(\mathbb{A}) \ | \  A \mbox{ is a finite support for } x \}.
\end{equation*}
We write $\supp{x}$ when $S$ is clear from the context.
 Clearly, each $a\in \mathbb{A}$ is finitely supported by $\{a\}$, therefore $\supp{a}=\{a\}$.

\subsection{The ``new'' quantifier}
\label{sec:new}
In Nominal Logic~\cite{pitts2003nominal}, the ``new'' quantifier $\new$ is used to quantify over new names. 
Given two elements $x$, $y$  of a nominal set, write $x  \#  y$ as an abbreviation of  $\supp{x} \cap \supp{y} = \emptyset$; then $\new a. P(a,x)$ (read ``for some/any new atom $a$, $P(a,x)$'') abbreviates $\forall a \in \mathbb{A}. a \# x \Rightarrow P(a,x)$ or equivalently $\exists a \in \mathbb{A}. a \# x \wedge  P(a,x)$ (see Theorem 3.9 in~\cite{pitts2013nominal} for more details).  

Intuitively, the ``new'' quantifier is used to indicate that a predicate holds for some (any) new atoms. The formula $\new a. \phi$ means ``for a fresh name $a$, $\phi$ holds''. Since there is an infinite supply of names and elements of nominal sets have a finite support, it is always possible to pick a fresh name. Moreover, it can be shown that if a property $\phi(a)$ holds for some fresh name $a$ then it holds for all fresh names, so the way in which the fresh atom is chosen is not important.

 In proof systems for Nominal Logic, several approaches have been proposed to define rules to introduce and eliminate the $\new$ quantifier (see, e.g.,~\cite{pitts2003nominal,Cheney2005, Cheney2008}). We follow Cheney's approach~\cite{Cheney2005}, where the formula $\new a. \phi$ (read ``$a$ is a \emph{newly quantified atom} in  $\phi$'') is well formed if $\phi$ is a well-formed formula and $a \in \mathbb{A}$ is semantically fresh. In other words, the introduction of a $\new$ quantifier for $a$ in a formula $\phi$ requires $a$ to be a new atom, fresh for all the variables in $\phi$ (see~\cite{Cheney2005}, Figure 5, where the well-formedness rules are given). We implicitly assume that $a$ is a newly generated atom when we introduce $\new$ in our formulas in the next section.
Operationally, proving a formula $\new a. \phi$ boils down to  picking a fresh atom $c$ and proving that $\phi$ holds when the occurrences of $a$ are replaced by the fresh atom $c$.


\section{Fixed-point Constraints}~\label{sec:constraints}
The native notion of equality on nominal terms is $\alpha$-equivalence, written $s\aleq t$. As mentioned in section~\ref{sec:alpha-fresh}, this relation is usually axiomatised using the \emph{freshness relation} between atoms and terms. The notion of freshness is derived from the notion of support, which in turn is defined using permutation fixed-points (see Definition~\ref{def:support}). We can define freshness using the quantifier $\new$  combined with a notion of fixed-point, as shown by Pitts~\cite{pitts2013nominal} (page 53):
$$a \# X\Leftrightarrow \new a'. \swap{a}{a'}\cdot X=X.$$

In this paper, instead of defining $\alpha$-equivalence using freshness, we define it using  the more primitive notion of \emph{fixed-point} under the action of permutations. We will denote this relation $\faleq$, and show that it coincides with $\aleq$ on ground terms, i.e., the relation  defined using permutation fixed points  corresponds to the relation defined using freshness. For non-ground terms, there is also a correspondence, but under different kinds of assumptions (fixed-point constraints vs.\ freshness constraints).

\subsection{Fixed-points of permutations and term equality}
\label{sec:fixpte}
We start by defining a binary relation that describes which elements of a nominal set $S$ are fixed-points of a permutation $\pi \in \permset$:

 \begin{defi}[Fixed-point relation]
 \label{def:fixed-point}
Let $S$ be a nominal set. The \emph{fixed-point relation} $\fixp \subseteq  \permset\!  \times\! S$ is defined as:
$\pi \fixp x \Leftrightarrow \pi \act x = x.$ 
Read ``$\pi\fixp x$'' as  \emph{``$\pi$ fixes $x$''}.
\end{defi}

It is easy to see that  if $\dom{\pi}\cap \supp{x}=\emptyset$ then $\pi \fixp x$ holds (this follows directly from the definition of support, see Definition~\ref{def:support}). However, the other direction does not hold in general. For example, consider expressions built using atoms and a binary commutative operator $+$; the permutation $\pi = (a\ b)$ fixes the equivalence class of the expression $a+b$ despite the fact that its support coincides with $\dom{\pi}$.
 
Permutation fixed-points will play an important role in the definition of $\alpha$-equality of nominal terms. Below we define    {\em fixed-point constraints} and {\em equality constraints} using  predicates $\fixp$  and   $\faleq$ and then give  deduction rules to derive fixed-point and equality judgements.

Intuitively, for $s$ and $t$ ground nominal terms
\begin{itemize}
\item $s\faleq t$ will mean that $s$ and $t$ are $\alpha$-equivalent, i.e., equivalent modulo renaming of abstracted atoms.
\item  $\pi \fixp t$ will mean that the permutation $\pi$ fixes the nominal term $t$, that is,  $\pi \cdot  t\faleq t$. This  means that $\pi$ has ``no effect'' on $t$ except for the renaming of bound names, for instance,  $\swap{a}{b}\fixp [a]a$ but not  $\swap{a}{b} \fixp \tf{f}\ a$. 
\end{itemize}

In the case of non-ground terms, a fixed-point or $\alpha$-equality constraint has to be evaluated in a context, which provides information about permutations that fix  the variables.

\begin{defi}[Fixed-point and equality constraints]
 A \emph{fixed-point constraint} is a pair $\pi\fixp t$ of a permutation $\pi$ and a term $t$.  An $\alpha$-{\em equivalence constraint} is a pair of the form $s\faleq t$. 
 
 We call a fixed-point constraint of the form $\pi\fixp X$ a \emph{primitive fixed-point constraint} and a finite set of such constraints is called  a {\em fixed-point context}. 
 $\Upsilon, \Psi, \ldots$ range over fixed-point  contexts. 
 
 We write $\pi \fixp \var{t}$ as an abbreviation for the set of constraints $\{\pi \fixp X \mid X \in \var{t}\}$.
\end{defi}

We now introduce some notation: 

The set $\var{\Upsilon}$ of variables  is defined as expected: it contains all the variables mentioned in the fixed-point context $\Upsilon$. 

The set of permutations of a fixed-point  context $\Upsilon$ with respect to the variable $X\in \var{\Upsilon}$, denoted by $\perm{\upsvar{X}}$, is defined as $\perm{\upsvar{X}}:=\{\pi\ | \ \pi\fixp X \in \Upsilon\}$.

For a substitution $\sigma$  and  a fixed-point  context $\Upsilon$  we define $\Upsilon\sigma:= \{\pi \fixp X\sigma \, | \, \pi \fixp X\in \Upsilon\}$.

To axiomatise the relation $\fixp$, we rely on the notion of \emph{conjugation} of permutations. 
Conjugacy plays an important role in group theory~\cite{hungerford73}. Recall that in a group $G$, the element $b$ is a conjugate of  $a$  if there exists an element $g$ such that $g a g^{-1} = b$. It can be easily shown that conjugacy is an equivalence relation: if $b$ is a conjugate of $a$ then $a$ is a conjugate of $b$. The conjugacy class of $a$ consists of the elements of the form $g a g^{-1}$, for $g\in G$. In the case of permutation groups, 
 the conjugate of $\pi$ with respect to $\rho$, denoted as $\pi^\rho$, is the result of the composition: $ \rho \circ \pi \circ \rho^{-1}$.
\[
\begin{array}{cccccccc}
\pi^\rho: & A &\overset{\rho^{- 1}}{\rightarrow} & A
& \overset{\pi}{\rightarrow}& A&\overset{\rho}{\rightarrow}& A\\
& a &\mapsto & \rho^{-1}(a) &\mapsto & \pi(\rho^{-1}(a)) &\mapsto & \rho(\pi(\rho^{-1}(a)))
\end{array}
\]


The notion of \emph{support} of a permutation or nominal term will be necessary for the next results. Considering permutations as elements of a nominal set, it follows from Definition~\ref{def:support} that  $\supp{\pi}=\dom{\pi}$; similarly, in the case of a ground term $t$, the support coincides with its set of free names, i.e., $\supp{t}=\fn{t}$. In the case a term $t$ is not ground, it should be considered in a  context, say $\Upsilon$, which gives information about its variables. We define it after introducing fixed-point and $\alpha$-equality judgements.

\begin{defi}[Judgements]
A \emph{fixed-point  judgement} is a tuple $\Upsilon \cent \new \overline{c}.\ \pi \fixp t$ of a fixed-point  context and a fixed-point  constraint where the atoms $\overline{c}$ are newly quantified. 
Here $ \overline{c}=c_1, \ldots, c_n$, and $n=0$ stands for empty quantification; in that case we may omit the $\new$ quantifier and write simply $\Upsilon \cent  \pi \fixp t$.

An \emph{$\alpha$-equivalence judgement} is a tuple 
$\Psi \cent  \new \overline{c}.\ s\faleq t$ of a fixed-point  context and an equality constraint
where the atoms $\overline{c}$ are newly quantified. 
\end{defi}
The derivable fixed-point  and $\alpha$-equivalence judgements are defined by the rules in  Figures~\ref{fig.rules.fixp} and \ref{fig.rules.equal}. 
We give an example before explaining the rules in more detail.

\begin{exa}\label{ex:fixedpoint}
The term $[a]\tf{f}a$ is a fixed-point  for the permutation $(a \ b)$, since $(a \ b)\act [a] \tf{f} a  = [b]\tf{f}b$, which is $\alpha$-equivalent to $[a]\tf{f}a$. As expected, we can derive   $\cent (a \ b) \fixp [a]\tf{f}a$. For this we use the rule \rulefont{\fixp abs}, which requires to prove $\cent \new c_1.(a \ b) \fixp (a\ c_1)\act \tf{f}a$. Since $(a\ c_1)\act \tf{f}a = \tf{f}c_1$, we need to prove  $\cent \new c_1.(a \ b) \fixp c_1$, which holds by \rulefont{\fixp a}.
However, $\tf{f}a$  is not a fixed-point  for $(a \ b)$: we cannot derive $\cent (a \ b) \fixp \tf{f}a$. This is to be expected, since $\tf{f}a$ and $\tf{f}b$ are not $\alpha$-equivalent; we cannot derive $\cent (a \ b) \act \tf{f}a \faleq \ \tf{f} a$. 
\end{exa}

\begin{figure*}[ht]
\hrule
\small{
$$
\begin{array}{l@{\hspace{1cm}}c} 
\begin{prooftree}
\pi(a) = a
\justifies 
\Upsilon\cent \new \overline{c}.\pi\fixp a
\using \rulefont{\fixp a}
\end{prooftree}
&
\begin{prooftree}
\supp{\pi^{{\pi'}^{-1}}}\setminus \{\overline{c}\}\subseteq \supp{\perm{\upsvar{X}}}
\justifies 
\Upsilon\cent \new \overline{c}.\pi\fixp\pi'\cdot X
\using \rulefont{\fixp var}
\end{prooftree}
\\[4ex]
\begin{prooftree}
\Upsilon\cent \new \overline{c}.\pi\fixp t
\justifies 
\Upsilon\cent \new \overline{c}.\pi\fixp \tf{f}\ t
\using \rulefont{\fixp \tf{f}} 
\end{prooftree}
&
\begin{prooftree}
\Upsilon\cent \new \overline{c}.\pi\fixp t_1  \quad \ldots \quad \Upsilon\cent \new \overline{c}.\pi\fixp t_n
\justifies
\Upsilon\cent \new\overline{c}.\pi\fixp  (t_1,\ldots, t_n)
\using \rulefont{\fixp tuple}
\end{prooftree}\\[4ex]
\begin{prooftree}
\Upsilon\cent \new \overline{c}, c_1.\pi \fixp   \swap{a}{c_1}\act t 
 \justifies 
\Upsilon\cent \new \overline{c}.\pi \fixp  [a]t
\using \rulefont{\fixp abs}  
\end{prooftree}
&
\end{array}
$$
}
\hrule
\caption{Fixed-point  rules.  }
\label{fig.rules.fixp}
\end{figure*}

The premise of rules \rulefont{\fixp abs} and \rulefont{\faleq\!{\tf ab}} introduces a newly quantified atom $c_1$. Recall that the introduction of a  $\new$ quantifier requires the use of a new atom by well formedness (see Section~\ref{sec:new}). Therefore, in rules  \rulefont{\fixp abs} and \rulefont{\faleq\!{\tf ab}} we can assume that the atom $c_1$ does not occur in the conclusion.

In rule $\rulefont{\fixp var}$, the condition 
$\supp{\pi^{{\pi'}^{-1}}}\setminus \{\overline{c}\}\subseteq \supp{\perm{\upsvar{X}}}$
means that the permutation $\pi^{{\pi'}^{-1}}$ can only affect atoms in  $\perm{\Upsilon|_X}$ and new atoms, hence it fixes $X$. 
Rules $\rulefont{\fixp {\tf f}}$ and $\rulefont{\fixp\texttt{tuple}}$ are straightforward.  Rule $\rulefont{\fixp abs}$ is the most interesting one. 
The intuition behind this rule is the following: $[a]t$ is a fixed-point of $\pi$ if $\pi\act [a]t$ is $\alpha$-equivalent to $[a]t$, that is, $[\pi(a)]\pi\act t$ is $\alpha$-equivalent to $[a]t$;
the latter means that the only free atom in $t$ that could be affected by $\pi$ is $a$, hence, if we  replace occurrences of $a$ in $t$ with another, new atom $c$, $\pi$ should have no effect.
Note that if a judgement $\Upsilon \cent \new \overline{c}. \pi \fixp t$ is derivable, then  so is $\Upsilon \cent \new \overline{c}, \overline{c'}. \pi \fixp t$ for a new set of atoms $\overline{c'}$ (there is an infinite supply of new atoms).


\begin{figure*}[ht]
\hrule
\small{
$$
\begin{array}{c@{\hspace{1cm}}c}
\begin{prooftree}
\justifies \Upsilon \cent \new \overline{c}.\ a \faleq a
\using \rulefont{\faleq\!a}
\end{prooftree}
&
\begin{prooftree}
\supp{ (\pi')^{-1}\circ \pi} \setminus \{\overline{c}\} \subseteq \supp{\perm{\upsvar{X}}} 
\justifies 
\Upsilon\cent  \new \overline{c}.\ \pi \cdot X \faleq \pi'\cdot X
\using \rulefont{\faleq\!var} 
\end{prooftree}
\\[4ex]
\begin{prooftree}
\Upsilon\cent \new \overline{c}.\ t\faleq t'
\justifies 
\Upsilon\cent  \new \overline{c}.\ \tf{f}\ t \faleq \tf{f} \ t'
\using \rulefont{\faleq\!\tf{f}} 
\end{prooftree}
&
\begin{prooftree}
\Upsilon\cent \new \overline{c}.t_1\faleq t_1' \quad \ldots \quad\Upsilon \cent \new \overline{c}.t_n\faleq t_n'
\justifies 
\Upsilon\cent  \new \overline{c}.(t_1,\ldots, t_n) \faleq (t_1',\ldots, t_n')
\using \rulefont{\faleq\!tuple} 
\end{prooftree}
\\[4ex]
\begin{prooftree}
\Upsilon\cent  \new \overline{c}.\ t\faleq t' 
\justifies 
\Upsilon\cent  \new \overline{c}.\ [a]t \faleq [a]t'
\using \rulefont{\faleq\![a]} 
\end{prooftree}
&
\begin{prooftree}
\Upsilon\cent \new \overline{c}.\ s\faleq (a \ b )\act t \qquad \Upsilon\cent \new \overline{c},c_1. (a\ c_1) \fixp t
\justifies
\Upsilon\cent \new \overline{c}.\ [a] s \faleq [b] t
\using \rulefont{\faleq\!{\tf ab}}
\end{prooftree}
\end{array}
$$}
\hrule
\caption{Rules for $\alpha$-equality.}
\label{fig.rules.equal}
\end{figure*}

The $\alpha$-equivalence relation is defined using fixed-point judgements (see rule $(\faleq{\tf{ab}})$). Rules $(\faleq {\tf a})$, $(\faleq {\tf f})$, $(\faleq {\tf{[a]}})$ and $(\faleq {\tf{tuple}})$ are defined as expected, whereas the intuition behind  rule $(\faleq  {\tf{ var}})$ is similar to the corresponding rule in Figure~\ref{fig.rules.fixp}.
The most interesting rule is $(\faleq{\tf{ab}})$. Intuitively, it states that for two abstractions $[a]s$ and $[b]t$ to be equivalent, we must obtain equivalent terms if we rename in one of them, in our case $t$, the abstracted atom $b$ to $a$, so that they both use the same atom. Moreover, the atom $a$ should not occur free in $t$, which is checked by stating that $\swap{a}{c_1}$ fixes $t$ for some newly generated atom $c_1$.

\begin{exa} Notice that $(b\,e)\fixp X \vdash  (a\, b) \fixp [a](\tf{f}a, (c\, d)\cdot X)$. This is derived as follows: by  rule $\rulefont{\fixp{\tf abs}}$ we have to show $(b\,e)\fixp X \vdash  \new g. (a\, b) \fixp (\tf{f}g, (a\,g)(c\, d)\cdot X)$;  by rule $\rulefont{\fixp{\tf tuple}}$ this requires   $(b\,e)\fixp X \vdash  \new g. (a\, b) \fixp \tf{f}g$ and $(b\,e)\fixp X \vdash  \new g. (a\, b) \fixp (a\,g)(c\, d)\cdot X$; the former holds by application of rules $\rulefont{\fixp\tf{f}}$  and $\rulefont{\fixp{\tf a}}$ since $(a\,b)\act g = g$, and the latter by rule $\rulefont{\fixp{\tf var}}$, since $\supp{(a\ b)^{(c\ d)(a \ g)} } \setminus \{g\}  = \{g, b\} \setminus \{g\} = \{b\}$ and $\supp{(b\ e)} = \{b,e\}$. 

However, it is not possible to derive $(a\,e)\fixp X \vdash  (a\, b) \fixp [a](\tf{f}a, (c\, d)\cdot X)$ since in this case $b$ is not in the support of the permutations that fix $X$ (only $(a\ e)$  fixes $X$ according to our context).

\end{exa}

We prove below that $\faleq$ is indeed an equivalence relation, for which we need to study  the properties of the relations $\faleq$ and $\fixp$, starting with \emph{inversion} and \emph{equivariance}.

\begin{lem}[Inversion]
The inference rules for $\faleq$ are invertible.
\end{lem}

\begin{proof}
Straightforward since  there is only one rule for each syntactic class of terms. Note that in the case of rule $\rulefont{\faleq\!{\tf ab}}$, the permutative convention ensures that $a$ and $b$ are different atoms. 
\end{proof}

Equivariance is an important notion in nominal logic: equivariant subsets of nominal sets are closed under permutation (see Definition 1.8 in~\cite{pitts2013nominal}).  Below we show that the fixed-point and $\alpha$-equivalence relations are equivariant. 

\begin{lem}[Equivariance]
\label{lem:equivariance}\hfill
\begin{enumerate}
\item 
$\Upsilon\cent \new \overline{c}.\pi \fixp \  t$  iff  $ \Upsilon  \cent\new \overline{c}. \pi^\rho \fixp \rho \act t$, for any permutation $\rho$.

\item 
$\Upsilon \cent \new \overline{c}.s\faleq t$ iff  $ \Upsilon \cent \new \overline{c}.\pi \act  s\faleq \pi \act t$, for any permutation $\pi$.
\end{enumerate}
\end{lem}

\begin{proof} 
For both parts the ``if'' statement is trivial, by taking the identity permutation, $\Id$. We prove the ``only if'' statements.
\begin{enumerate}
\item 
The proof is by rule induction using the rules in Figure~\ref{fig.rules.fixp}. We consider cases depending on the last rule applied in the derivation. 

\begin{enumerate}
\item The rule applied is $\rulefont{\fixp a}$.

In this case, $t=a$. Notice that, 
$\Upsilon\cent\pi'\fixp a $ if, and only if, $\pi'(a)=a$. We want to show that $\Upsilon\cent \new \overline{c}.\pi^\rho \fixp \rho (a)$, for any permutation $\rho$.  Notice that  $\pi^\rho(\rho(a)) = (\rho \circ \pi \circ \rho^{-1}) (\rho(a))= (\rho \circ \pi) (a)= \rho(\pi(a))=\rho(a)$, and the result follows. 
\item The rule applied is $\rulefont{\fixp var}$.

In this case one has  $t= \pi_1\cdot X$ and $\Upsilon\cent \new \overline{c}.\pi\fixp \pi_1\cdot X$. 
We need to prove that 
$\Upsilon\cent \new\overline{c}.\pi^\rho\fixp  \rho \act (\pi_1 \cdot X)$, 
 which requires 
 $\supp{(\rho \circ \pi_1)^{-1}\circ \pi^\rho \circ (\rho \circ \pi_1)}\setminus\{\overline{c}\} \subseteq \supp{\perm{\upsvar{X}}}$.

Note that, since $\Upsilon \cent\new \overline{c}.\pi \fixp \pi_1\cdot X $, it follows that $\supp{\pi_1^{-1} \circ \pi  \circ \pi_1}\setminus\{\overline{c}\} \subseteq \supp{\perm{\Upsilon|_X}}$. 
By definition of $\pi^\rho$, it follows that 
\begin{equation*}
\begin{aligned}
\supp{(\rho \circ \pi_1)^{-1}\circ \pi^\rho \circ (\rho \circ \pi_1)}&= \supp{(\rho \circ \pi_1)^{-1}\circ (\rho \circ \pi\circ \rho^{-1}) \circ (\rho \circ \pi_1)}\\
&=\supp{(\pi_1^{-1}\circ \rho^{-1})}\circ (\rho \circ \pi\circ \rho^{-1}) \circ (\rho \circ \pi_1)\\
&=\supp{\pi_1^{-1}\circ \pi\circ \pi_1}\\
\end{aligned}
\end{equation*}
Therefore, the result follows.

\item The rule applied is  $\rulefont{\fixp abs}$.

In this case, $t=[a]t'$ and $\Upsilon \cent \new\overline{c}.\pi \fixp [a]t'$, therefore, there exists a proof $\Pi$ :

\begin{prooftree}
\[\Pi 
\justifies 
\Upsilon \cent  \new \overline{c}, c_1. \  \pi \fixp \swap{a}{c_1} \cdot t' \]
\justifies
\Upsilon \cent  \new  \overline{c}. \ \pi \fixp [a]t'
\using \rulefont{\fixp abs}
\end{prooftree}

By induction hypothesis, for an arbitrary permutation $\rho$, there exists a derivation $\Pi'$ for : 
$$\Upsilon\cent \new \overline{c},c_1. \ \pi^\rho \fixp \rho \act (\swap{a}{c_1}\act t').$$

Since $\rho \act (\swap{a}{c_1}\act t')= \swap{\rho(a)}{\rho(c_1)}\act (\rho \act t')= \swap{\rho(a)}{c_1}\act (\rho \act t')$ because $c_1$ is a new name, we have

\begin{prooftree}
\[
\Pi'
\justifies
\Upsilon \cent \new \overline{c},c_1. \ \pi^\rho\fixp \swap{\rho(a)}{c_1}\act (\rho\act t')\] 
\justifies
\Upsilon \cent \new \overline{c}. \  \pi^\rho \fixp [\rho(a)]\rho \act t'
\end{prooftree}

and the result follows.

\item The rule applied is $\rulefont{\fixp \tf{f}}$ or  $\rulefont{\fixp tuple}$.

These cases follow directly by induction hypothesis.
\end{enumerate}

\item 
The proof is by rule induction using  the rules in Figure~\ref{fig.rules.equal}.
 
\begin{enumerate}
\item The rule applied is $\rulefont{\faleq a}$.

In this case, $s\faleq t$ is an instance of $a \faleq a$. It is straightforward to check that $\Upsilon\cent \new \overline{c}.\pi \cdot a \faleq \pi \cdot a$, via an application of the rule $\rulefont{\faleq a}$.

\item The rule applied is $\rulefont{\faleq var}$.

In this case, $\Upsilon\cent \new \overline{c}.\ \pi_1 \cdot X \faleq \pi_2 \cdot X$ and $\supp{\pi_2^{-1} \circ \pi_1}\setminus\{\overline{c}\} \subseteq \supp{\perm{\upsvar{X}}}$.

Notice that $\supp{(\pi \circ \pi_2)^{-1} \circ \pi \circ \pi_1}= \supp{\pi_2^{-1}\circ \pi_1}$. Therefore, by rule $\rulefont{\faleq var}$,
$\Upsilon\cent \new \overline{c}.\ \pi \cdot( \pi_1 \cdot X ) \faleq \pi \cdot (\pi_2 \cdot X)$.

\item The cases corresponding to rule $\rulefont{\faleq f}$ and rule $\rulefont{\faleq [a]}$ follow directly by induction hypothesis.












\item The rule applied is $\rulefont{\faleq ab}$.

In this case, there exist derivations $\Pi_1$ and $\Pi_2$ such that

\begin{prooftree}
\[\Pi_1
\justifies
\Upsilon \cent \new \overline{c}.\ t\faleq \swap{a}{b}\cdot s\] \qquad \[\Pi_2
\justifies
\Upsilon \cent \new \overline{c},c_1. \ \swap{a}{c_1}\fixp s\]
\justifies
\Upsilon \cent \new \overline{c}.\ [a]t \faleq [b]s
\using \rulefont{\faleq ab}
\end{prooftree}

By induction hypothesis, and part (1) of this lemma:
\begin{enumerate}
\item there exists a proof $\Pi_1'$ of  $\Upsilon \cent \new \overline{c}.\ \pi \act t\faleq \pi \act (\swap{a}{b}\cdot s)$.
\item there exists a proof $\Pi_2'$ of $\Upsilon\cent \new \overline{c}, c_1. \ \swap{a}{c_1}^\pi\fixp \pi \act s$.
\end{enumerate}

It is easy to show that $\pi \act (\swap{a}{b} \act s)=\swap{\pi(a)}{\pi(b)}\act(\pi\act s)$, and since $c_1$ is new, $\pi(c_1) = c_1$, therefore, the result follows. 

\begin{prooftree}
\[\Pi_1'
\justifies
\Upsilon \cent \new \overline{c}.\ \pi\act t\faleq \swap{\pi(a)}{\pi(b)}\cdot \pi \act s\] \qquad \[\Pi_2'
\justifies
\Upsilon \cent \new \overline{c}, c_1. \ \swap{\pi(a)}{c_1}\fixp \pi \act s\]
\justifies
\Upsilon \cent \new \overline{c}.\ [\pi(a)]\pi \cdot t \faleq [\pi(b)]\pi \cdot s
\using \rulefont{\faleq ab}
\end{prooftree}

\end{enumerate}
\end{enumerate}
\end{proof}


\begin{exa}
Notice that $(a \ c) \fixp X \cent (a \ b ) \fixp (b \ c)\cdot X$, for 
$\swap{a}{c} \fixp X \cent \swap{a}{c}\fixp  X$ and $\swap{a}{b}^{\swap{b}{c}}  = \swap{a}{c}$. Using the  Equivariance Lemma, with $\rho = (b\ c)$, we deduce:
\begin{equation}
\begin{aligned}
\swap{a}{c} \fixp X &\cent \swap{a}{b}^ {\swap{b}{c}} \fixp X\Leftrightarrow (a \ c) \fixp X \cent (a \ b ) \fixp (b \ c)\cdot X \\
\end{aligned}
\end{equation}
\end{exa}


\begin{proposition}[Strengthening for $\fixp$]\label{prop:streng.fixp}\mbox{}\\
If $\Upsilon, \pi \fixp X \cent \new \overline{c}. \ \pi' \fixp s$ and $\supp{\pi }\subseteq \supp{\perm{\upsvar{X}}}$ or $X\notin \var{s}$ then $ \Upsilon \cent \new \overline{c}.\pi' \fixp s$.
\end{proposition}

\begin{proof}
By induction on the structure of $s$.
\begin{itemize}
\item $s= a$

In this case, $\Upsilon, \pi \fixp X \cent \new \overline{c}. \ \pi' \fixp a$. 
By rule $\rulefont{\fixp a}$, 

\begin{equation*}
\begin{aligned}
\Upsilon, \pi \fixp X \cent \new \overline{c}. \  \pi' \fixp a & \Leftrightarrow \pi'(a)=a
 \Leftrightarrow  \Upsilon\cent \new \overline{c}. \ \pi' \fixp a
\end{aligned}
\end{equation*}

\item $s=\pi_1 \cdot X$ (the case $ s = \pi_1 \cdot Y$ is trivial).

In this case,  
$\Upsilon, \pi \fixp X \cent \new \overline{c}. \ \pi' \fixp \pi_1 \cdot X $.

\begin{equation*}
\begin{aligned}
\Upsilon, \pi \fixp X \cent \new \overline{c}. \ \pi' \fixp \pi_1\!\cdot\! X &\Leftrightarrow \supp{(\pi')^{\pi_1^{-1}}}\setminus\{\overline{c}\}\subseteq\supp{\perm{\upsvar{X}}\cup \{\pi\}}  \;\; (rule ~ \rulefont{\fixp var})\\
&\Leftrightarrow \supp{(\pi')^{\pi_1^{-1}}}\setminus\{\overline{c}\}  \subseteq\supp{\perm{\upsvar{X}}} \qquad (by ~ assumption)\\
&\Leftrightarrow \Upsilon \cent \new \overline{c}.\pi' \fixp \pi_1 \cdot X \hspace{4cm} (by ~ rule ~ \rulefont{\fixp var})
\end{aligned}
\end{equation*}

\item $s = [a]s'$

In this case, $\Upsilon, \pi \fixp X \cent \new \overline{c}. \ \pi' \fixp [a]s'$, $\supp{\pi}\subseteq \supp{\perm{\upsvar{X}}}$ or $X\notin \var{s'}$. Then there exists a derivation 

\vspace{2mm}

\begin{prooftree} 
\[\mathcal{D}
\justifies 
\Upsilon, \pi\fixp X \cent \new \overline{c}, c_1. \ \pi'\fixp \swap{a}{c_1} \act s'\]
\justifies
\Upsilon, \pi\fixp X \cent \new \overline{c}. \ \pi' \fixp [a]s'
\using \rulefont{\fixp abs}
\end{prooftree}

\vspace{2mm}

By induction hypothesis, there exists a derivation $\mathcal{D}'$ such that

\vspace{2mm}

\begin{prooftree} 
\[\mathcal{D'}
\justifies 
\Upsilon\cent \new \overline{c},c_1. \ \pi'\fixp \swap{a}{c_1}\act s'\]
\justifies
\Upsilon \cent \pi' \fixp [a]s'
\using \rulefont{\fixp abs}
\end{prooftree}

\vspace{2mm}

 \item $s = \tf{f}s'$ or $s=(s_1,\ldots s_n)$

These cases follow easily by induction hypothesis.
\qedhere

\end{itemize}
\end{proof}

\begin{proposition}[Strengthening for $\faleq$]\label{prop:streng.faleq}\mbox{}\\
If $\Upsilon, \pi \fixp X \cent \new \overline{c}.\ s\faleq t$ and 
$\supp{\pi}\subseteq  \supp{\perm{\upsvar{X}}}$ 
or $X\not \in \var{s,t}$, then 
$ \Upsilon \cent \new \overline{c}.\ s\faleq t$.
\end{proposition}
\begin{proof}
The proof is similar to that of Proposition~\ref{prop:streng.fixp} and omitted.
\end{proof}

The following auxiliary lemma permits to deduce  a fixed-point constraint for a permutation $\pi$ and term $t$ from constraints involving $t$. Intuitively, it states that we can deduce that $\pi$ fixes $t$ if the support of $\pi$ is contained in permutations that fix $t$. 

\begin{lem}
\label{lem:fixp-smallerpi}
Let $I$ be a set of indices.
If $\Upsilon \cent \new \overline{c_i}. \pi_i \fixp t$ for all $i \in I$, and $\pi$ is a permutation such that 
$\supp{\pi}\setminus \{\overline{c}\} \subseteq \bigcup_{i\in I}\supp{\pi_i} \setminus \bigcup_{i\in I} \overline{c_i}$ then $\Upsilon \cent \new \overline{c}. \pi \fixp t$.
\end{lem}

\begin{proof} By induction on $t$.
\begin{itemize}
    \item If $t = a$ then by assumption  $\pi_i(a) =a$ for all $i \in I$ since $\Upsilon \cent \new \overline{c_i}. \pi_i \fixp a$ (rule \rulefont{\fixp a}). Then,   $\pi(a) = a$
    by the assumption on the support of $\pi$. The result follows by rule \rulefont{\fixp a}. 
    \item If $t = \pi'\act X$ then $\supp{\pi_i^{{\pi'}^{-1}}}\setminus \{\overline{c_i}\}  \subseteq \supp{\perm{\Upsilon|_X}}$ 
    because $\Upsilon \cent \new \overline{c_i}. \pi_i \fixp t$ for all $i \in I$ (rule \rulefont{\fixp var}). 
    
   Since $\supp{\pi}\setminus \{\overline{c}\} \subseteq \bigcup_{i\in I}\supp{\pi_i} \setminus \bigcup_{i\in I} \overline{c_i}$,  also   $\supp{\pi^{{\pi'}^{-1}}}\setminus \{\overline{c}\} \subseteq \bigcup_{i\in I}\supp{\pi_i^{{\pi'}^{-1}}} \setminus \bigcup_{i\in I} \overline{c_i}$. The result follows by rule \rulefont{\fixp var}.
  
  \item The other cases follow directly by induction.
    \qedhere
\end{itemize}
\end{proof}

The following correctness property states that  $\fixp$ is indeed the fixed-point relation.

\begin{thm}[Correctness]
\label{th:fix.alpha}
Let $\Upsilon$,  $\pi$  and $t$ be a fixed-point context, a permutation and a nominal term, respectively.  $\Upsilon \cent\new \overline{c}.\ \pi \fixp t$ ~~~ iff ~~~ $\Upsilon \cent  \new \overline{c}.\ \pi \cdot t \faleq t$.

\end{thm}

\begin{proof}
In both directions the proof is by induction on the structure of the term $t$, distinguishing cases according to the  last rule applied in the derivation.

\noindent ($\Rightarrow$) 
The interesting cases correspond to variables and abstractions; the other cases follow directly by the induction hypothesis.
\begin{enumerate}
%
%
\item The last rule is $\rulefont{\fixp var}$.

In this case, $t=\pi'\cdot X$ and $\supp{(\pi')^{-1}\circ \pi \circ \pi'}\setminus \{\overline{c}\}\subseteq \supp{\perm{\Upsilon|_X}}$, and therefore  $\Upsilon\cent \new \overline{c}.\ \pi\act (\pi' \cdot X) \faleq \pi' \cdot X$, via rule $\rulefont{\faleq var}$.

\item The last rule is $\rulefont{\fixp abs}$.
In this case,  $t=[a]t'$ and by assumption there is a derivation of the form:

\begin{prooftree}
\[\Pi
\justifies
\Upsilon\cent \new \overline{c}, c_1. \ \pi \fixp\swap{a}{c_1}\act t'\]
\justifies 
\Upsilon\cent \new\overline{c}.\pi\fixp [a]t'
\end{prooftree}

We need to prove that  $\Upsilon \cent \new \overline{c}.\ [\pi(a)]\pi\act t' \faleq [a]t'$, that is, 
$\Upsilon \cent \new \overline{c}.\ \pi\act t' \faleq \swap{\pi(a)}{a} \act t'$ and also 
$\Upsilon \cent \new \overline{c}, c_1.\swap{\pi(a)}{c_1} \fixp t'$.

By IH, there exists a proof $\Pi'$ for $\Upsilon \cent \new \overline{c}, c_1.\pi \act(\swap{a}{c_1}.t') \faleq \swap{a}{c_1}\act t'$. 
The following equivalence holds, since $c_1$ is newly quantified:
\begin{equation}\label{eq:ind.hyp}
\begin{aligned}
\Upsilon \cent \new \overline{c}, c_1.\pi \cdot(\swap{a}{c_1}\act t') \faleq \swap{a}{c_1}\act t'&\iff \Upsilon\cent \new \overline{c}, c_1.\swap{\pi(a)}{c_1}\act (\pi \act t') \faleq \swap{a}{c_1} \act t'\\
\end{aligned}
\end{equation}
Also, $\Upsilon\cent \new \overline{c}, c_1.(\pi \act t') \faleq \swap{\pi(a)}{c_1} \act (\swap{a}{c_1}\act t')$ by Equivariance.  Note that  $ \swap{\pi(a)}{c_1} \act (\swap{a}{c_1}\act t')= \swap{\pi(a)}{a} \act t'$, hence $\Upsilon \cent \new \overline{c}.\pi\act t' \faleq \swap{\pi(a)}{a} \act t'$ as required.

Since $\Upsilon\cent \new \overline{c}, c_1. \ \pi \fixp\swap{a}{c_1}\act t'$, we also have 
$\Upsilon\cent \new \overline{c}, c_1. \ \pi^{(a\ c_1)} \fixp t'$ by equivariance 
and we deduce
$\Upsilon \cent \new \overline{c}, c_1.\swap{\pi(a)}{c_1} \fixp t'$ as required.

\end{enumerate}

\noindent ($\Leftarrow$) Similarly,  we proceed by induction on the derivation of $\Upsilon \cent \new \overline{c}.\ \pi\act t\faleq t$. We show the interesting cases, the others follow directly by induction.

\begin{enumerate}
\item The rule is $\rulefont{\faleq var}$.

In this case $t=\pi_1\cdot X$ and there exists a proof of $\Upsilon\cent\new \overline{c}.\ \pi \cdot(\pi_1\cdot X)\faleq \pi_1\cdot X$. Therefore, $\supp{\pi_1^{-1} \circ \pi \circ \pi_1}\setminus\{\overline{c}\} \subseteq \supp{\perm{\Upsilon|_X}}$, and one can conclude that $\Upsilon \cent \new \overline{c}.\ \pi\fixp \pi_1\cdot X$, via application of rule $\rulefont{\fixp var}$.








\item The rule is $\rulefont{\faleq ab}$.

In this case, $t = [a]t'$ and 
since $\Upsilon \cent \new \overline{c}.\ [\pi(a)]\pi\act t' \faleq [a]t'$, we know (Inversion):

$\Upsilon \cent \new \overline{c}.\ \pi\act t' \faleq \swap{\pi(a)}{a}\act t'$ and $\Upsilon \cent \new \overline{c}, c_1. \swap{\pi(a)}{c_1}\fixp t'$. 

By Equivariance $\Upsilon \cent \new \overline{c}.(\swap{\pi(a)}{a}\circ\pi)\act t' \faleq  t'$ and by induction
$\Upsilon \cent \new \overline{c}. (\swap{\pi(a)}{a}\circ\pi)\fixp t'$.

From $\Upsilon \cent \new \overline{c}, c_1. \swap{\pi(a)}{c_1}\fixp t'$ and
$\Upsilon \cent \new \overline{c}. (\swap{\pi(a)}{a}\circ\pi)\fixp t'$ we deduce 
$\Upsilon \cent \new \overline{c},c_1. \pi^{(a \ c_1)} \fixp t'$ by Lemma~\ref{lem:fixp-smallerpi},
since $\supp{\pi^{(a \ c_1)}} \setminus \{\overline{c},c_1\} \subseteq (\supp{\swap{\pi(a)}{a}\circ\pi}\setminus \{\overline{c}\}) \cup (\supp{(\pi(a)\ c_1)} \setminus \{c_1\})$. By Equivariance, $\Upsilon \cent \new \overline{c},c_1. \pi \fixp (a \ c_1)\act t'$, and therefore $\Upsilon \cent \new \overline{c}.\ \pi \fixp [a]t'$ by  $\rulefont{\fixp ab}$, as required.
%
\qedhere
\end{enumerate}
\end{proof}

Recall that $\Upsilon\sigma$ denotes the set of fixed-point constraints obtained by applying the substitution $\sigma$ to the constraints in $\Upsilon$ (see Section~\ref{sec:fixpte}).
Below we abbreviate $\Upsilon \cent \pi_1 \fixp t_1$, \ldots, $\Upsilon \cent \pi_n \fixp t_n$ 
as  $\Upsilon \cent \pi_1 \fixp t_1, \ldots,  \pi_n \fixp t_n$. Thus, $\Upsilon \cent \Upsilon' \sigma$ means that each of the constraints in $\Upsilon' \sigma$ is derivable from $\Upsilon$.

\begin{proposition}[Preservation under Substitution]\label{prop:weak}
 Suppose that $\Upsilon\cent \Upsilon' \sigma$. Then, 
\begin{enumerate}
\item $\Upsilon'\cent \new \overline{c}. \ \pi \fixp s \Longrightarrow \Upsilon \cent \new \overline{c}. \ \pi \fixp s\sigma$.
\item  $\Upsilon'\cent  \new \overline{c}.\ s\faleq t \Longrightarrow \Upsilon \cent \new \overline{c}.\ s\sigma \faleq t\sigma$.
\end{enumerate}
\end{proposition}

\begin{proof}
By induction on the rules in Figures~\ref{fig.rules.fixp} and ~\ref{fig.rules.equal}. 

\begin{enumerate}
\item We distinguish cases depending on  the last rule applied in the derivation of $\Upsilon' \cent \new \overline{c}. \pi \fixp s.$
\begin{enumerate}
\item The rule is \rulefont{\fixp a}.

In this case, $s=a$ and
 
\begin{prooftree}
\pi(a)= a
\justifies
\Upsilon'\cent \new \overline{c}.\pi \fixp a
\using \rulefont{\fixp a}
\end{prooftree}

\

The result follows trivially, since $a\sigma=a$ for any substitution $\sigma$.
\item The rule is \rulefont{\fixp f}

In this case, $s=\tf{f} s'$ and there exists a proof $ \Pi'$ of

\begin{prooftree}
\[ \Pi'
\justifies
\Upsilon' \cent \new \overline{c}. \pi \fixp s'\]
\justifies
\Upsilon' \cent \new \overline{c}. \pi \fixp \tf{f} s'
\using 
\rulefont{\fixp f}
\end{prooftree}

\

By induction hypothesis, there exists a proof $\Pi^{''}$ such that 

\begin{prooftree}
\[ \Pi^{''}
\justifies
\Upsilon \cent \new \overline{c}. \pi \fixp s'\sigma\]
\justifies
\Upsilon \cent \new \overline{c}. \pi \fixp (\tf{f}s')\sigma
\using 
\rulefont{\fixp f}
\end{prooftree}


\item The rule is \rulefont{\fixp abs}

In this case $s=[a]s'$ and there exists a proof $\Pi'$ of the form 

\begin{prooftree}
\[ \Pi'
\justifies
\Upsilon' \cent \new \overline{c},c_1. \pi \fixp \swap{a}{c_1}\cdot s'\]
\justifies
\Upsilon' \cent \new \overline{c}. \pi \fixp [a]s'
\using 
\rulefont{\fixp abs}
\end{prooftree}

By induction hypothesis, there exists a proof $\Pi^{''}$ of the form

\begin{prooftree}
 \Pi
\justifies
\Upsilon \cent \new \overline{c},c_1. \pi \fixp (\swap{a}{c_1}\cdot s')\sigma
\end{prooftree}

Since $(\swap{a}{c_1}\cdot s')\sigma = \swap{a}{c_1}\cdot (s'\sigma)$ and $[a](s'\sigma)=([a]s')\sigma$, it follows that 

\begin{prooftree}
 \[\Pi
\justifies
\Upsilon \cent \new \overline{c},c_1. \pi \fixp \swap{a}{c_1}\cdot (s'\sigma)\]
\justifies
\Upsilon \cent \new \overline{c}. \pi \fixp [a](s'\sigma)
\using 
\rulefont{\fixp abs}
\end{prooftree}

\item The rule is \rulefont{\fixp tuple}

This case is analogous to the case  for \rulefont{\fixp f}, and follows directly by IH.

\item The rule is \rulefont{\fixp var}

In this case, $s=\rho\cdot X$ and $\Upsilon' \cent \new \overline{c}.\pi \fixp \rho\cdot X$ holds, that is, $\supp{\pi^{\rho^{-1}}}\setminus \{\overline{c}\}\subseteq \supp{\perm{\Upsilon'|_X}}$.

From $\Upsilon\cent \Upsilon'\sigma$, it follows that, $\Upsilon \cent \pi_1\fixp X\sigma$, for all $\pi_1\fixp X \in \Upsilon'$. 

Therefore, $\Upsilon \cent \new \overline{c}.\ \pi^{\rho^{-1}} \fixp X\sigma $ by Lemma~\ref{lem:fixp-smallerpi},
and the result follows by Equivariance.

\end{enumerate}

\item We proceed by analysing the last rule used in the derivation of $\Upsilon' \cent \new \overline{c}.\ s\faleq t.$
\begin{enumerate}
    \item The last rule is \rulefont{\faleq a}.
    
    This case is trivial.
    
    \item The last rule is \rulefont{\faleq var}.
    
    In this case 
    we have $\Upsilon' \cent \new \overline{c}.\ \pi \cdot X \faleq \pi' \cdot X$ and therefore
    $\supp{(\pi')^{-1}\circ \pi}\setminus \{\overline{c}\} \subseteq \supp{\perm{\Upsilon'|_X}}$.

    From $\Upsilon\cent \Upsilon'\sigma$, it follows that $\Upsilon \cent \pi_1\fixp X\sigma$, for all $\pi_1\fixp X \in \Upsilon'$. 
    
    Therefore, $\Upsilon \cent \new \overline{c}. (\pi')^{-1}\circ \pi \fixp X\sigma $ by Lemma~\ref{lem:fixp-smallerpi}, and the result follows by Theorem~\ref{th:fix.alpha} and Equivariance.

    \item The last rule is \rulefont{\faleq f} or \rulefont{\faleq [a]}.
    These cases follow directly by induction.
    
    \item The last rule is \rulefont{\faleq ab}.

    In this case, we know $\Upsilon \cent \new \overline{c}.\ [a]s \faleq [b]t$ and therefore (by Inversion) $\Upsilon \cent \new \overline{c}.\ s \faleq \swap{a}{b}\act t$ and $\Upsilon \cent \new \overline{c},  c_1. \swap{a}{c_1}\fixp t$.
    
    By induction hypothesis, $\Upsilon \cent \new \overline{c}.\ s\sigma \faleq \swap{a}{b}\act t\sigma$, and by part (1) of this proposition, $\Upsilon \cent \new \overline{c},  c_1.\swap{a}{c_1}\fixp t\sigma$. The result then follows using rule \rulefont{\faleq ab}.
\qedhere    
\end{enumerate}
\end{enumerate}
\end{proof}

\begin{corollary}[Weakening] Assume $\Upsilon' \subseteq \Upsilon$.
\begin{enumerate}
\item $\Upsilon'\cent \new \overline{c}. \ \pi \fixp s \Longrightarrow \Upsilon \cent \new \overline{c}. \ \pi \fixp s$.
\item  $\Upsilon'\cent \new \overline{c}.\ s\faleq t \Longrightarrow \Upsilon \cent \new \overline{c}.\ s \faleq t$.
\end{enumerate}
\end{corollary}

\subsection{Alternative approaches to new name generation}
\label{sec:alt.to.new}
In the previous section we used the ``new'' quantifier to deal with  new names in constraints: judgements involve newly quantified constraints, and  
in rules \rulefont{\fixp ab} and \rulefont{\faleq ab}, when new fresh names are needed, a newly quantified atom is used. 

An alternative approach consists of quantifying judgements instead of constraints.
More precisely, we can define judgements of the form $\new \overline{c}.(\Upsilon \cent \pi \fixp t)$ and $\new \overline{c}.(\Upsilon \cent s \faleq t)$.  Rules to derive this kind of judgements can be easily obtained by adapting the rules in Figures~\ref{fig.rules.fixp} and \ref{fig.rules.equal}. We show the rules for fixed-point judgements in Figure~\ref{fig.rules.fixp.alt}. For the sake of readabilty, we omit brackets around judgements in the rules and  assume the quantifiers have maximal scope.

\begin{figure*}[ht]
\hrule
\small{
$$
\begin{array}{l@{\hspace{1cm}}c} 
\begin{prooftree}
\pi(a) = a
\justifies 
\new \overline{c}. \Upsilon\cent \pi\fixp a
\using \rulefont{\fixp a}
\end{prooftree}
&
\begin{prooftree}
\supp{\pi^{{\pi'}^{-1}}}\setminus \{\overline{c}\}\subseteq \supp{\perm{\upsvar{X}}}
\justifies 
\new \overline{c}.\Upsilon\cent \pi\fixp\pi'\cdot X
\using \rulefont{\fixp var}
\end{prooftree}
\\[4ex]
\begin{prooftree}
\new \overline{c}.\Upsilon\cent \pi\fixp t
\justifies 
\new \overline{c}.\Upsilon\cent \pi\fixp \tf{f}\ t
\using \rulefont{\fixp \tf{f}} 
\end{prooftree}
&
\begin{prooftree}
\new \overline{c}.\Upsilon\cent \pi\fixp t_1  \quad \ldots \quad \new \overline{c}.\Upsilon\cent \pi\fixp t_n
\justifies
\new \overline{c}.\Upsilon\cent \pi\fixp  (t_1,\ldots, t_n)
\using \rulefont{\fixp tuple}
\end{prooftree}\\[4ex]
\begin{prooftree}
\new \overline{c},c_1.\Upsilon\cent \pi \fixp   \swap{a}{c_1}\act t 
 \justifies 
\new \overline{c}.\Upsilon\cent \pi \fixp  [a]t
\using \rulefont{\fixp abs}  
\end{prooftree}
&
\end{array}
$$
}
\hrule
\caption{Alternative fixed-point  rules.  }
\label{fig.rules.fixp.alt}
\end{figure*}

Yet another  approach consists of using a name generator whenever new names are required (see~\cite{Cheney2008,Cheney2017}), as done in~\cite{DBLP:conf/rta/Ayala-RinconFN18}.
In this case, the new quantifier is not needed, and new names are generated dynamically by using an external generator, under the assumption that the generator outputs a new (unused name) whenever needed. For comparison, we recall in Figure~\ref{fig.rules.fixp.gen} the rules given  in~\cite{DBLP:conf/rta/Ayala-RinconFN18}
to derive fixed-point judgements $\Upsilon \cent \pi \fixp t$. Note that in rule \rulefont{\fixp abs} the fixed-point context $\Upsilon$ is augmented with constraints  $(c_1\ c_2)\fixp \var{t}$.  These constraints serve to store the information about the fact that the generated atoms are  "new". The trick is to generate two atoms,  even though only one new atom is needed ($c_1$, to replace $a$), in order to be able to express the fact that $c_1$ is new. This solution is inspired by  the link between freshness, the new quantifier and fixed-point equations,  $a \# X\Leftrightarrow \new a'. \swap{a}{a'}\cdot X=X$, mentioned at the beginning of Section~\ref{sec:constraints}.

\begin{figure*}[ht]
\hrule
\small{
$$
\begin{array}{l@{\hspace{1cm}}c} 
\begin{prooftree}
\pi(a) = a
\justifies 
\Upsilon\cent \pi\fixp a
\using \rulefont{\fixp a}
\end{prooftree}
&
\begin{prooftree}
\supp{\pi^{{\pi'}^{-1}}}\subseteq \supp{\perm{\upsvar{X}}}
\justifies 
\Upsilon\cent \pi\fixp\pi'\cdot X
\using \rulefont{\fixp var}
\end{prooftree}
\\[4ex]
\begin{prooftree}
\Upsilon\cent \pi\fixp t
\justifies 
\Upsilon\cent \pi\fixp \tf{f}\ t
\using \rulefont{\fixp \tf{f}} 
\end{prooftree}
&
\begin{prooftree}
\Upsilon\cent \pi\fixp t_1  \quad \ldots \quad \Upsilon\cent \pi\fixp t_n
\justifies
\Upsilon\cent \pi\fixp  (t_1,\ldots, t_n)
\using \rulefont{\fixp tuple}
\end{prooftree}
\end{array}
$$
$$
\begin{prooftree}
\Upsilon, (c_1\ c_2)\fixp \var{t} \cent \pi \fixp   \swap{a}{c_1}\act t 
 \justifies 
\Upsilon\cent \pi \fixp  [a]t
\using \rulefont{\fixp abs}  ~~where ~c_1 ~and ~c_2 ~ are ~ new ~ names
\end{prooftree}
$$
}
\hrule
\caption{Fixed-point rules using a name generator.  }
\label{fig.rules.fixp.gen}
\end{figure*}

The rules to derive $\alpha$-equivalence judgements following this approach are recalled in Figure~\ref{fig.rules.equal.gen}.

\begin{figure*}[ht]
\hrule
\small{
$$
\begin{array}{cc}
\begin{prooftree}
\justifies \Upsilon \cent a \faleq a
\using \rulefont{\faleq\!a}
\end{prooftree}
&
\begin{prooftree}
\supp{ (\pi')^{-1}\circ \pi} \subseteq \supp{\perm{\upsvar{X}}} 
\justifies 
\Upsilon\cent  \pi \cdot X \faleq \pi'\cdot X
\using \rulefont{\faleq\!var} 
\end{prooftree}
\\[4ex]
\begin{prooftree}
\Upsilon\cent t\faleq t'
\justifies 
\Upsilon\cent  \tf{f}\ t \faleq \tf{f} \ t'
\using \rulefont{\faleq\!\tf{f}} 
\end{prooftree}
&
\begin{prooftree}
\Upsilon\cent t_1\faleq t_1' \quad \ldots \quad\Upsilon \cent t_n\faleq t_n'
\justifies 
\Upsilon\cent  (t_1,\ldots, t_n) \faleq (t_1',\ldots, t_n')
\using \rulefont{\faleq\!tuple} 
\end{prooftree}
\\[4ex]
\begin{prooftree}
\Upsilon\cent  t\faleq t' 
\justifies 
\Upsilon\cent  [a]t \faleq [a]t'
\using \rulefont{\faleq\![a]} 
\end{prooftree}
&
\begin{prooftree}
\Upsilon\cent s\faleq (a \ b )\act t \quad \Upsilon, (c_1 \ c_2) \fixp \var{t} \cent  (a\ c_1) \fixp t
\justifies
\Upsilon\cent [a] s \faleq [b] t
\using \rulefont{\faleq\!{\tf ab}} 
\end{prooftree}
\end{array}
$$
}
\hrule
\caption{Equality rules using a name generator. In rule $\rulefont{\faleq\!{\tf ab}}$, $c_1$ and $c_2$ are new names.}
\label{fig.rules.equal.gen}
\end{figure*}

The latter approach has the advantage of having a simpler syntax for judgements (without quantification) but relies on an external name generator, which  can be seen as  adding a form of state. In other words, in the latter approach  rules produce side-effects. Although not as elegant as the approaches using $\new$, this approach is convenient from an implementation point of view. In the next sections we use this technique to translate primitive freshness constraints to primitive fixed-point constraints, and to specify a unification algorithm.


\subsection{From freshness to fixed-point constraints and back again}
\label{sec:fresh.to.fixp}

In this section we show that the $\alpha$-equivalence relation defined in terms of {\em freshness constraints}, denoted as $\aleq$, is equivalent to $\faleq$, given that a transformation $[\_]^\fixp $ from primitive freshness to primitive fixed-point constraints and a transformation $[\_]^\#$ from primitive fixed-point to primitive freshness constraints  can be defined. 

Below we denote by $\mathfrak{F}_\#$ the family of freshness contexts, and by $\mathfrak{F}_{\fixp}$ the family of fixed-point contexts. 
The mapping $[\_]_\fixp$ associates each  primitive freshness constraint  with a fixed-point constraint; it extends to freshness contexts in the natural way.
\[
\begin{array}{cccl}
[\_ ]_\fixp: & \mathfrak{F}_{\#} &\longrightarrow &\mathfrak{F}_{\fixp}\\
&a \# X& \mapsto & \swap{a}{c_a}\fixp X  \mbox{ where } c_a \mbox{ is a new name}.
\end{array}
\]
We denote by $[\Delta]_\fixp$ the image of $\Delta$ under $[\_]_\fixp$.

The mapping $[\_]_\#$  associates each primitive fixed-point constraint  with a freshness constraint; it extends to fixed-point contexts in the natural way.
\[
\begin{array}{cccl}
[\_ ]_\#: & \mathfrak{F}_{\fixp} &\longrightarrow &\mathfrak{F}_{\#}\\
&\pi \fixp X& \mapsto & \supp{\pi}\# X.
\end{array}
\]
We denote by $[\Upsilon]_\#$ the image of $\Upsilon$ under $[\_]_\#$.

Below we abbreviate the set of constraints 
$\{a_1\# t, \ldots a_n\# t \mid a_1, \ldots, a_n \in A\}$ as $\overline{A \# t}$. 

Using the above-specified translation of primitive constraints, we translate freshness judgements into fixed-point judgements with newly quantified constraints:

\centerline{$\Delta\cent a\# t$ is translated as  $[\Delta]_\fixp\cent \new c.\swap{a}{c}\fixp t$}

\noindent
and vice-versa:

\centerline{$\Upsilon\cent \new \overline{c}. \pi\fixp t$ is translated as $[\Upsilon]_\#\cent \overline{\supp{\pi}\setminus\{\overline{c}\} \# t}$.}

\begin{thm}\label{lem:fresh.fixp} \leavevmode
\begin{enumerate}
\item $\Delta\cent a\# t \Leftrightarrow [\Delta]_\fixp\cent \new c.\swap{a}{c}\fixp t$.
\item $\Upsilon\cent \new \overline{c}. \pi\fixp t \Leftrightarrow [\Upsilon]_\#\cent \overline{\supp{\pi}\setminus\{\overline{c}\} \# t}$.
\end{enumerate}
\end{thm}

\begin{proof}
Part (1):

$(\Rightarrow)$ By induction on the derivation of $\Delta\cent a\# t$ (see the rules in Figure~\ref{fig.rules.freshness}).
We distinguish cases based on the last rule used in the derivation.
\begin{itemize}
    \item 
    If the rule is \rulefont{\# a} then $t$ is an atom $b$ (it cannot be $a$), and the result follows by rule \rulefont{\fixp a}.
\item
If the  rule is \rulefont{\# var} then $t = \pi\act X$ and  since $\Delta \cent a \#  \pi\act X$ we know $\pi^{-1}(a) \# X \in \Delta$ by Inversion. Hence, 
$\swap{\pi^{-1}(a)}{c_{\pi^{-1}(a)}} \fixp X \in [\Delta]_\fixp$, and therefore
$\supp{\swap{a}{c}^{\pi^{-1}}} \setminus \{c\} \subseteq \supp{\perm{[\Delta]_\fixp|_X}}$ (recall that $\pi^{-1}(c) = c$ since $c$ is a new atom).
The result then follows by rule \rulefont{\fixp var}.
\item
The cases for \rulefont{\# tuple} and \rulefont{\# f} follow directly by induction.
\item
If the rule is \rulefont{\#[a]} then $t = [a]s$ and we need to prove $[\Delta]_\fixp\cent \new c.\swap{a}{c}\fixp [a]s$. By rule  \rulefont{\fixp abs}, it suffices to show $[\Delta]_\fixp\cent \new c, c_1.\swap{a}{c}\fixp \swap{a}{c_1}\act s$. By Equivariance, this is equivalent to $[\Delta]_\fixp\cent \new c, c_1.\swap{a}{c}^{\swap{a}{c_1}}\fixp s$, or equivalently $[\Delta]_\fixp\cent \new c, c_1.\swap{c_1}{c}\fixp s$, which holds trivially since $c$ and $c_1$ are new atoms ($\swap{c_1}{c}\act s = s$) 
\item
If the rule is \rulefont{\# abs} then  $t = [b]s$. By assumption, 
 $\Delta \cent a \# [b]s$, hence $\Delta \cent a \# s$.  By induction hypothesis, 
 $[\Delta]_\fixp\cent \new c.\swap{a}{c}\fixp  s$, and by Equivariance, $[\Delta]_\fixp\cent \new c, c_1.\swap{a}{c}\fixp \swap{b}{c_1}\act s$.  The result follows by rule \rulefont{\fixp abs}.
\end{itemize}

$(\Leftarrow)$ By induction on the derivation of  $[\Delta]_\fixp\cent \new c.\swap{a}{c}\fixp t$. We distinguish cases based on the last rule used in the derivation.
\begin{itemize}
    \item 
    If the rule is \rulefont{\fixp a} then $t$ is an atom $b$ (it cannot be $a$), and the result follows by rule \rulefont{\# a}.
\item
If the  rule is \rulefont{\fixp var} then $t = \pi\act X$ and  since $[\Delta]_\fixp\cent \new c.\swap{a}{c}\fixp \pi\act X$ we know $\supp{\swap{a}{c}^{\pi^{-1}}} \setminus \{c\}  \subseteq \supp{\perm{[\Delta]_\fixp|_X}}$ by Inversion. Hence, 
$\pi^{-1}(a)\in \supp{\perm{[\Delta]_\fixp|_X}}$, and therefore
$\pi^{-1}(a)\# X \in \Delta$ by definition of the mapping $[\cdot]_\fixp$.
The result then follows by rule \rulefont{\# var}.
\item
The cases for \rulefont{\fixp tuple} and \rulefont{\fixp f} follow directly by induction.
\item
If the rule is \rulefont{\fixp abs} then 
there are two cases, $t =[a]s$ and $t = [b]s$. \\
If  $t =[a]s$ the result follows directly by rule \rulefont{\#[a]}.\\
If  $t =[b]s$ then  by assumption and Inversion, $[\Delta]_\fixp\cent \new c c_1.\swap{a}{c}\fixp \swap{b}{c_1}\act s$. By Equivariance, $[\Delta]_\fixp\cent \new c c_1.\swap{a}{c}\fixp  s$. By induction hypothesis we deduce 
 $\Delta \cent a \# s$, and the result follows by rule \rulefont{\# abs}.
\end{itemize}
Part (2):

$(\Rightarrow)$ By induction on the derivation of $\Upsilon\cent \new \overline{c}. \pi\fixp t$. Again we distinguish cases based on the last rule used in the derivation.
The only interesting cases are  rule \rulefont{\fixp var} and \rulefont{\fixp abs}.
\begin{itemize}
    \item If the last rule applied is \rulefont{\fixp var} then $t = \pi'\act X$. By Inversion, $\supp{\pi^{{\pi'}^{-1}}} \setminus \{\overline{c}\} \subseteq \supp{\perm{\Upsilon|_X}}$.
    Therefore, for any $a \in \supp{\pi} \setminus \{\overline{c}\}$,  $\pi'^{-1}(a) \in \supp{\perm{\Upsilon|_X}}$. By definition of the mapping $[\cdot]_\#$, $\pi'^{-1}(a) \# X \in [\Upsilon]_\#$, and the result follows by rule \rulefont{\# var}.
    \item If the last rule applied is \rulefont{\fixp abs}  then $t = [a]s$. In this case, we know $\Upsilon \cent \new c.c_1. \pi \fixp \swap{a}{c_1}\act s$. By induction, 
    $[\Upsilon]_\# \cent \supp{\pi} \setminus \{c,c_1\}  \# \swap{a}{c_1}\act s$, hence 
    $[\Upsilon]_\# \cent \supp{\pi} \setminus \{a\}  \#  s$, and the result follows by rules \rulefont{\# [a]} and \rulefont{\# abs}. 
\end{itemize}

$(\Leftarrow)$ By induction on $t$. 

\begin{itemize}
\item If $t$ is an atom $a$ then $a \not\in \supp{\pi} \setminus \{\overline{c}\}$ (by Inversion, rule \rulefont{\#  a}). The result follows by rule \rulefont{\fixp a}.
    \item If  $t = \pi'\act X$, then the  rule applied in the derivation is  rule \rulefont{\#  var}. By Inversion, for any atom $a$ in $\supp{\pi} \setminus \{\overline{c}\}$,  $\supp{\pi^{\pi'^{-1}}} \subseteq  \supp{\perm{\Upsilon|_X}}$ and the result follows by rule \rulefont{\fixp var}.
    \item If $t = [a]s$ then for any atom $b\in \supp{\pi} \setminus \{\overline{c}\}$ ($b$ different from $a$ by the permutative convention), we know  $[\Upsilon]_\# \cent b \# s$, since by assumption  $[\Upsilon]_\# \cent b \# [a]s$ and the last rule applied in  the derivation must have been  rule \rulefont{\# abs}.  In particular, for any atom $b \in \supp{\pi}\setminus\{\overline{c},a,c_1\}$,  $[\Upsilon]_\# \cent b \# s$. Then, by induction hypothesis, $\Upsilon \cent \new c.c_1. \pi^{\swap{a}{c_1} }\fixp s$
     and the result follows by Equivariance and rule \rulefont{\fixp abs}. 
    \item The other cases follow directly by induction.
      \qedhere
\end{itemize}
\end{proof}

\begin{thm}
\label{th:aleq.to.faleq}
$\faleq$ coincides with $\aleq$ on ground terms, that is, $\cent s \aleq t \iff  \cent s \faleq t$.
More generally, 
\begin{enumerate}
\item $\Delta \cent s\aleq t \Rightarrow [\Delta]_\fixp\cent s\faleq t$.
\item 
$\Upsilon \cent \new\overline{c}.\ s\faleq t \Rightarrow [\Upsilon]_\# , \Delta \cent s\aleq t$, where $\Delta \cent \overline{c}\# \var{s, t}$.
\end{enumerate}
\end{thm}

\begin{proof}
\begin{enumerate}
\item The first part is proved  by induction on the derivation of $\Delta\cent s\aleq t$, distinguishing cases according to the last rule applied (see Figure~\ref{fig.rules.equaf}). The interesting cases correspond to  $\rulefont{\aleq var}$ and  $\rulefont{\aleq ab}$.
\begin{itemize}
    \item 
If the last rule applied is $\rulefont{\aleq var}$:

\begin{prooftree}
\diffs{\pi}{\pi_1}\# X \subseteq \Delta
\justifies
\Delta \cent \pi \cdot X \aleq \pi_1 \cdot X
\using \rulefont{\aleq var}
\end{prooftree}

We want to show that $[\Delta]_\fixp \cent \pi \cdot X \faleq \pi_1\cdot X$. To use  rule $\rulefont{\faleq var}$, we need to show that $\supp{\pi_1^{-1}\circ \pi}\subseteq \supp{\perm{([\Delta]\fixp)|_X}}$. 
Let $b \in \supp{\pi_1^{-1}\circ \pi}$ and suppose $b\notin \diffs{\pi}{\pi_1}$. Then $\pi(b)= \pi_1(b)$ and $\pi_1^{-1}(\pi(b)) = b$, contradiction. Therefore, $b\in \diffs{\pi}{\pi_1}$ and $\swap{b}{c_b}\fixp X \in [\Delta]_\fixp$ (for $c_b$ a new name), and the result follows.
\item
If the last rule applied is 
$\rulefont{\aleq ab}$ the result follows directly by induction and Theorem~\ref{lem:fresh.fixp}.
\end{itemize}

\item The second part is proved  by induction on the derivation of $\Upsilon \cent s\faleq t$, distinguishing cases according to the last rule applied. Again the interesting cases correspond to  variables and abstractions. The proof follows the lines of the previous part and is omitted.
  \qedhere
\end{enumerate}
\end{proof}

As a corollary, since $\aleq$ is an equivalence relation~\cite{Urban2004}, we deduce that $\faleq$ is also an equivalence relation.
\begin{thm}
$\faleq$ is an equivalence relation.
\end{thm}

\begin{lem}[$\fixp$ preservation under $\faleq$]
\label{lem:fixp-preserves-faleq}

If $\Upsilon\cent s\faleq t$ and $\Upsilon \cent \new \overline{c}.\pi \fixp s$ then $\Upsilon \cent \new \overline{c}.\pi \fixp t$.
\end{lem}

\begin{proof}
Direct consequence of Theorem~\ref{th:fix.alpha}, Equivariance and Transitivity.
\end{proof}
%
%
%
%
%
%
Having proved that $\faleq$ is an equivalence relation, and that  $\fixp$ is correctly defined with respect to $\faleq$, we can use $\fixp$ to define the support of a non-ground term.
We denote the support of a term in context $\Upsilon \cent t$ as $\suppt{\Upsilon}{t}$. As indicated in Definition~\ref{def:support}, the set of atoms in the support of an element of a nominal set can be characterised by using permutation fixed points. In the particular case of terms, the previous results justify the definition of $\suppt{\Upsilon}{t}$  using fixed-point judgements as follows.

\begin{defi} 
\label{def:supp-non-ground}
Let $\Upsilon \vdash t$ be a term in context. 
The {\em support} of $t$ with respect to $\Upsilon$, $\suppt{\Upsilon}{t}$, is the smallest set $A$ of atoms such that for any permutation $\pi$, 
$$(\forall a \in A, \pi(a) = a) \Rightarrow \Upsilon \vdash \pi \fixp t.$$
\end{defi}

As expected, $\alpha$-equivalent terms have the same support.

\begin{lem}
\label{lem:supp-fixp-faleq}\hfill
\begin{enumerate}
 \item If $\Upsilon \cent  s \faleq  t$ then $\suppt{\Upsilon}{s} = \suppt{\Upsilon}{t}$.
 \item   $\Upsilon \cent \new \overline{c}. \pi \fixp t$  
 if and only if 
 $\supp{\pi}\setminus{\{\overline{c}\}} \cap \suppt{\Upsilon}{t} = \emptyset$. 
\end{enumerate}
\end{lem}

\begin{proof} 
Direct consequence of Definition~\ref{def:supp-non-ground} and Lemma~\ref{lem:fixp-preserves-faleq}.
\end{proof}

\section{Nominal Unification via fixed-point  constraints}
\label{sec:unif}
 In this section we address the problem of unifying nominal terms. Solutions for unification problems will be represented using fixed-point constraints and substitutions. After defining unification problems, we present  a simplification algorithm that computes the most general unifier for a unification problem, provided the problem has a solution (otherwise the algorithm stops indicating that there is no solution). 
 To specify the nominal unification algorithm, in this section we follow the approach to dealing with new atoms that relies on a name generator.
 
 \begin{defi}
 A \emph{unification problem} $\probl$ consists of a finite set of equality and fixed-point constraints of the form $s\faleq^? t$ and 
$\pi \fixp^? t$, 
 respectively. 
 \end{defi}
 Below we call 
 $\pi \fixp^? X$ 
 a \textbf{primitive} constraint.
 
 \begin{defi}[Solution]
\label{def:sol.unif}
A \emph{solution} for a problem ${\probl}$ is a pair of the form $\pair{\Phi}{\sigma}$ where the following conditions are satisfied:
\begin{enumerate}
\item 
$\Phi\cent \pi\fixp t\sigma$, if $\pi \fixp^? t\in \probl$;
\item $\Phi \cent s\sigma \faleq t\sigma$, if $s\faleq^? t \in \probl$.
\item $X\sigma \faleq  X\sigma \sigma $ for all $X\in \var{\probl}$ (the substitution is idempotent).
\end{enumerate} 
The \emph{solution set} for a problem $\probl$ is denoted by $\mathcal{U}(\probl)$.
 \end{defi}
 
 Solutions in $\mathcal{U}(\probl)$ can be compared using the following ordering.

\begin{defi}
\label{def.instantiation.ordering}
Let $\Phi_1,\Phi_2$ be fixed-point contexts, and $\sigma_1,\sigma_2$ substitutions.
Then $\pair{\Phi_1}{\sigma_1}\leq  \pair{\Phi_2}{\sigma_2}$ if there exists some 
$\sigma'$ such that 
$$
\text{for all $X$},\quad \Phi_2\cent X\sigma_1\sigma'\faleq X\sigma_2 
\quad\text{and}\quad
\Phi_2\cent \Phi_1\sigma' .
$$
\end{defi}


\begin{defi}
A \textbf{principal} (or \textbf{most general}) solution to a problem $\probl$ is a least element of $\mathcal{U}(\probl)$.
\end{defi}

 We design a unification algorithm via the simplification rules presented in Table~\ref{table:simpl.rules}. 
 These rules act on unification problems $\probl$ by transforming constraints into simpler ones, or instantiating variables in the case of rules $(\aleq inst1)$ and $(\aleq inst1)$. We call the latter \emph{instantiating rules}.
 We  abbreviate $(t_1,\ldots, t_n)$ as $(\widetilde{t})_{1..n}$, and for a set $S$, 
 $\overline{\pi \fixp S}= \{\pi \fixp X \ | \  X\in S\}$.

\begin{table}[ht]
\hrule
\small{
\[
\begin{array}{ll@{\hspace{1mm}}l@{\hspace{-2mm}}l}
(\fixp at) &\probl \uplus \{\pi \fixp^? a\} &\Longrightarrow& \probl, \mbox{ if }\ \pi(a) =a\\
(\fixp f)& \probl\uplus \{\pi \fixp^? 	 \tf{f}t\} &\Longrightarrow &\probl\cup \{\pi \fixp^? t\} \\
(\fixp t)&\probl\uplus \{\pi \fixp^? 	 (\widetilde{t})_n\}&\Longrightarrow &\probl\cup\{\pi \fixp^?t_1, \ldots, \pi\fixp^? t_n\}\\
(\fixp abs)& \probl\uplus \{\pi \fixp^? [a]t\} &\Longrightarrow &\probl\cup 
\{\pi \fixp^? \swap{a}{c_1} \act t, \overline{\swap{c_1}{c_2}\fixp^? \var{t}}\}\\
(\fixp var)& \probl\uplus \{\pi \fixp^? \pi'\cdot X\}&\Longrightarrow & \probl \cup 
\{\pi^{(\pi')^{-1}}\fixp^? X\}, \mbox{ if } \pi' \neq Id\\

(\faleq a)&\probl\uplus \{a\faleq^? a\}&\Longrightarrow & {\probl}\\
(\faleq f)&\probl\uplus \{\tf{f} \ t\faleq^? \tf{f} \ t'\}&\Longrightarrow & \probl\cup \{t\aleq^? t'\}\\
(\faleq t)&\probl\uplus \{(\widetilde{t})_n\aleq^? (\widetilde{t'})_n\} &\Longrightarrow& {\probl\cup \{t_1\faleq^? t'_1, \ldots, t_n\faleq^? t'_n\}}\\
(\faleq abs1)&\probl\uplus \{[a]t\faleq^? [a]t'\}&\Longrightarrow & {\probl\cup\{t\faleq^? t'\}}\\
(\faleq abs2)&\probl\uplus \{[a]t\faleq^? [b]s\} &\Longrightarrow &\probl\cup \{t\faleq^? \swap{a}{b}\act s, \swap{a}{c_1}\fixp^? s, \overline{\swap{c_1}{c_2}\fixp^? \var{s}}\}\\
(\faleq var)&\probl\uplus \{\pi \cdot X\faleq^? \pi'\act X\} &\Longrightarrow & \probl\ \cup \{(\pi')^{-1}\circ \pi \fixp^? X\}\\
(\faleq inst1)& \probl\uplus \{\pi \cdot X\faleq^? t\} &\hspace{-4mm}\stackrel{[X \mapsto\pi^{-1}.t]}\Longrightarrow & \hspace{4mm}\probl\{X \mapsto\pi^{-1}.t\}, \mbox{ if } X\notin \var{t}\\ 
(\faleq inst2)& \probl\uplus \{t \faleq^? \pi \cdot X\} &\hspace{-4mm}\stackrel{[X \mapsto\pi^{-1}.t]}\Longrightarrow & \hspace{4mm}\probl\{X \mapsto\pi^{-1}.t\}, \mbox{ if } X\notin \var{t}\\ 
\end{array}
\]}
\hrule
\caption{Simplification Rules. In  $(\fixp abs)$ and $(\faleq abs2 )$, $c_1$ and $c_2$ are new names.}
\label{table:simpl.rules}
\end{table}

 We write $\probl\Longrightarrow \probl'$ when $\probl'$ is obtained from $\probl$ by applying a simplification rule from Table~\ref{table:simpl.rules}, and we write $\stackrel{*}{\Longrightarrow}$ for the reflexive and transitive closure of $\Longrightarrow$.

 \begin{lem}[Termination]
	There is no infinite chain of reductions $\Longrightarrow$ starting from a problem $ \probl$.
 \end{lem}
 
 \begin{proof}
 Termination of the simplification rules follows directly from the fact that the following measure of the size of $\probl$ is  strictly decreasing:

 $[\probl] = (n_1,M)$ where $n_1$ is the number of different variables 
 used in $\probl$,
 and $M$ is the multiset of sizes of equality constraints and non-primitive fixed-point constraints occurring in $\probl$. 
 To compare $[\probl]$ and $[\probl']$ we use the lexicographic combination of the usual order on natural numbers, $>$,  and its multiset extension. We denote the order by $>_{lex}$. Thus, $[\probl] >_{lex} [\probl']$ if $\probl'$ has less variables than $\probl$, or if it has the same number of variables but smaller equality constraints and non-primitive fixed-point constraints.
 
Each simplification step $\probl\Longrightarrow \probl'$ either eliminates one variable (when an instantiation rule is used) and therefore decreases the first component of the interpretation, or leaves the first component unchanged but replaces a constraint with smaller ones and/or primitive ones (when a non-instantiating rule is used). Therefore, $\probl\Longrightarrow \probl'$ implies $[\probl] >_{lex} [\probl']$. Hence, it is not possible to have an infinite descending chain.
 \end{proof}

If $\probl\Longrightarrow^* \probl'$ and $\probl'$ is irreducible, we say that $\probl'$ is a normal form. We will show next that if instantiation rules are not used, each problem $\probl$ has a unique normal form. Indeed unicity of normal forms for non-instantiating rules is a consequence of the following property.

\begin{lem}[Confluence]
The relation $\Longrightarrow$ defined by the rules in Table~\ref{table:simpl.rules} without $(\aleq inst1)$ and $(\aleq inst2)$ is confluent.
\end{lem}

\begin{proof}
Confluence follows from the fact that the rules have no critical pairs (there are only trivial overlaps) and are terminating (by Newman's lemma).
\end{proof}

If instantiating rules are used, normal forms are not necessarily unique.  For example, the problem $(a\ b) \fixp^? X, X \faleq^? Y$ has two normal forms:
$$(a\ b) \fixp^? X, X \faleq^? Y \Longrightarrow (a\ b) \fixp^? Y$$ 
$$(a\ b) \fixp^? X, X \faleq^? Y \Longrightarrow (a\ b) \fixp^? X$$ 
However reductions do always terminate with some normal form, 
and all the normal forms of a problem are equivalent in a natural and useful sense (see Definition~\ref{def:sol.unif.computed} and Theorem~\ref{thm.sol.is.principal}). We will  use the notation  $\nftriple{\probl}$ to refer to any normal form of $\probl$.

 We say that an equality constraint $s\faleq^? t$ is \emph{reduced} when one of the following holds:
 \begin{enumerate}
 \item $s=a$ and $t=b$ are distinct atoms;
 \item $s$ and $t$ are headed with different function symbols, that is, $s= \tf{f}\  s'$ and $t= \tf{g}\ t'$;
 \item $s$ and $t$ have different term constructors, that is, $s=[a]s'$ and $t=\tf{f} \ t'$, for some term former $\tf{f}$, or $s=\pi \cdot X$ and $t=a$, etc. 
 \end{enumerate}
 A fixed-point constraint 
 $\pi \fixp^? s$
 is \emph{reduced} when it is of the form 
 $\pi\fixp^? a$ 
 and $\pi(a)\neq a$, or 
 $\pi \fixp^?X$; 
 the former is  \emph{inconsistent} whereas the latter is  \emph{consistent}.

\begin{exa}\label{ex:unif.fix}
For $\probl = [a]\tf{f}(X, a)\faleq^? [b]\tf{f} (\swap{b}{c}\cdot W, \swap{a}{c}\cdot Y)$, we obtain the following derivation chain using the rules in Table~\ref{table:simpl.rules}:
\begin{equation*}
\begin{aligned}
&[a]\tf{f}(X, a)\faleq^? [b]\tf{f}(\swap{b}{c}\cdot W, \swap{a}{c}\cdot Y) \Longrightarrow
\left\{\hspace{-1.5mm}
\begin{array}{l}
\tf{f}(X, a)\faleq^? \tf{f}(\swap{a}{b}\circ \swap{b}{c}\cdot W, \swap{a}{b}\circ\swap{a}{c}\cdot Y),\\
\swap{a}{c_1}\fixp^?\tf{f}(\swap{b}{c}\cdot W, \swap{a}{c}\act Y), \\
\swap{c_1}{c_2}\fixp^? W,  \swap{c_1}{c_2}\fixp^?Y
\end{array}\hspace{-2mm}
\right\}\\
&\stackrel{*}{\Longrightarrow}
\left\{
\begin{array}{l}
X\faleq^? \swap{a}{b}\circ \swap{b}{c}\cdot W,  a \faleq^?\swap{a}{b}\circ\swap{a}{c}\cdot Y,\\
\swap{a}{c_1}\fixp^?\swap{b}{c}\cdot W, \swap{a}{c_1}\fixp^? \swap{a}{c}\cdot Y,\swap{c_1}{c_2}\fixp^? W,  \swap{c_1}{c_2}\fixp^?Y
\end{array}
\right\}\\
&\stackrel{[Y\mapsto b]}{\Longrightarrow}
\left\{
\begin{array}{l}
X\faleq^? \swap{a}{b}\circ \swap{b}{c}\cdot W, \swap{a}{c_1}^{\swap{b}{c}} \fixp^? W,
\swap{a}{c_1}\fixp^? b, \swap{c_1}{c_2}\fixp^? W,  \swap{c_1}{c_2}\fixp^? b
\end{array}
\right\}\\
&\stackrel{*}{\Longrightarrow}
\left\{
\begin{array}{l}
X\faleq^? \swap{a}{b}\circ \swap{b}{c}\cdot W, \swap{a}{c_1} \fixp^? W, \swap{c_1}{c_2}\fixp^? W
\end{array}
\right\}\\
&\stackrel{[X\mapsto \swap{a}{b}\circ \swap{b}{c}\cdot W ]}{\Longrightarrow}
\left\{
\begin{array}{l}
\swap{a}{c_1} \fixp^? W,  \swap{c_1}{c_2}\fixp^? W
\end{array}
\right\} = \nftriple{\probl}.
\end{aligned}
\end{equation*}
\end{exa}

\begin{defi}[Characterisation of normal forms]
\label{def:successfulnf}
Let $\probl$ be a problem such that $\nftriple{\probl}={\probl'}$. We say that
$\probl'$ is \emph{reduced} when it consists of reduced constraints, and  \emph{successful} when $\probl' = \emptyset$ or contains only consistent reduced fixed-point constraints; otherwise,  $\nftriple{\probl}$\emph{ fails}.   
\end{defi}

The simplification rules (Table~\ref{table:simpl.rules}) specify a \emph{unification algorithm}:  we apply the simplification rules in a problem $\probl$ until we reach a normal form $\nftriple{\probl}$. 
\begin{defi}[Computed Solutions]
\label{def:sol.unif.computed}
If $\nftriple{\probl}$ \emph{fails} or contains reduced equational constraints, we say that $\probl$ is {\em unsolvable}; otherwise, $\nftriple{\probl}$ is \emph{solvable} and its solution, denoted $\sol{\probl}$, consists of the composition $\sigma$ of substitutions applied through the simplification steps and the fixed-point context 
$\Phi= \{\pi \fixp X \ | \pi \fixp^? X \in \nftriple{\probl}\}$.
\end{defi}

\begin{exa}[Continuing example~\ref{ex:unif.fix}]\label{ex:sol.unif}
Notice that
$\pair{\Psi}{\sigma}$, where 
$\Psi=\{\swap{a}{c_1}\fixp W, \swap{c_1}{c_2}\fixp W\}$ 
and $\sigma = \{Y\mapsto b, X \mapsto \swap{a}{b}\circ\swap{b}{c}\cdot W\}$, is a solution for $\probl$.
\end{exa}

We will now show that every solvable unification problem has a
 principal solution, computed by the unification algorithm, such that any other solution can be obtained as an instance of the principal one.
 
We start by proving that the non-instantiating rules preserve the set of  solutions.

\begin{lem}[Correctness of Non-Instantiating rules] 
\label{lem.preservation.of.solutions}
Let $\probl$ be a unification problem such that $\probl\Longrightarrow^*\probl'$
without using instantiating rules $(\faleq inst1)$ and $(\faleq inst2)$ then
\begin{enumerate}
\item
$\mathcal{U}(\probl)=\mathcal{U}(\probl')$, and 
\item
if $\probl'$ contains equational or inconsistent reduced fixed-point constraints then $\mathcal{U}(\probl)=\emptyset$.
\end{enumerate}
\end{lem}

\begin{proof}
The proof is by induction on the length of the derivation $\probl \overset{n}{\Longrightarrow} \probl'$.

\noindent{\bf Base Case.} $n=0$.
Then $\probl=\probl'$ and the result is trivial.

\noindent {\bf Induction Step.} Suppose, $n>0$ and consider the reduction chain 
$$\probl= \probl_1  \Longrightarrow \ldots \Longrightarrow \probl_{n-1}\Longrightarrow \probl_n=\probl'.$$ 
The proof follows by case analysis on the last rule applied in $\probl_{n-1}$. 
\begin{enumerate}
\item The rule is $(\fixp at)$.

In this case, 
$\probl_{n-1}=\probl_{n-1}'  \uplus \{ \pi \fixp^? a \}\Longrightarrow \probl_{n-1}'= \probl_n$, and $\pi(a)=a$. 

Let $\pair{\Psi}{\sigma}\in \mathcal{U}(\probl_{n-1})$, then 
\begin{enumerate}
  \item 
  $ \  \Psi \cent \pi' \fixp t\sigma$, for all $\pi'\fixp^?t\in \probl'_{n-1}$
 \item $ \  \Psi \cent t\sigma\faleq s\sigma$, for all $t\faleq^?s\in \probl_{n-1}'$;
 \item \ $X\sigma = X\sigma\sigma$, for all $X\in \var{\probl_{n-1}'}$.
 \end{enumerate}
 Therefore, $\pair{\Psi}{\sigma}\in \mathcal{U}(\probl_n)$ and $\mathcal{U}(\probl_{n-1})\subseteq \mathcal{U}(\probl_n)$. The other inclusion is trivial.
 
\item The rule is $(\fixp var)$.

In this case, 
$\probl_{n-1}=\probl_{n-1}'  \uplus \{ \pi \fixp^? \pi'\cdot X\}\Longrightarrow \probl_{n-1}'\cup \{\pi^{(\pi')^{-1}}\fixp^? X\}= \probl_n$, and $\pi'\neq Id$. 

Let $\pair{\Psi}{\sigma}\in \mathcal{U}(\probl_{n-1})$, then 
\begin{enumerate}
  \item 
  $ \  \Psi \cent \pi' \fixp t\sigma$, for all $\pi'\fixp^?t\in \probl'_{n-1}$, and $\Psi \cent  \pi \fixp \pi'\cdot X\sigma$.
 \item $ \  \Psi \cent t\sigma\faleq s\sigma$, for all $t\faleq^?s\in \probl_{n-1}'$;
 \item \ $X\sigma = X\sigma\sigma$, for all $X\in \var{\probl_{n-1}'}$.
 \end{enumerate}
Notice that 
\[
\begin{array}{ll}
\Psi \cent \pi \fixp \pi'\cdot X\sigma &\Rightarrow 
\Psi \cent \pi \act (\pi' \act X\sigma) \faleq (\pi' \cdot X\sigma), \mbox{ hence } \\
&\Psi \cent (\pi')^{-1} \circ \pi \circ \pi' \act (X\sigma) \faleq X\sigma \mbox{ by Equivariance (Lemma~\ref{lem:equivariance})}\\
&\Rightarrow \Psi \cent \pi^{(\pi')^{-1}} \fixp  X\sigma.
\end{array}
\]
Therefore, $\pair{\Psi}{\sigma}\in \mathcal{U}(\probl_{n})$ and $\mathcal{U}(\probl_{n-1})\subseteq \mathcal{U}(\probl_n)$. The other inclusion is similar.

\item The rule is $\fixp abs$.
Then
 \begin{equation*}
 \begin{aligned}
 \probl_{n-1}&= \probl'\uplus \{\pi \fixp^? [a]s\} \Longrightarrow \probl'\cup \{\pi \fixp^? \swap{a}{c_1}.s, \overline{(c_1\ c_2) \fixp \var{s}}\}=\probl_n.
\end{aligned}
\end{equation*}
where $c_1$ and $c_2$ are new names not occurring anywhere in the problem.

Let $\pair{\Psi}{\sigma}\in \mathcal{U}(\probl_{n-1}$) be a solution for $\probl_{n-1}$:
\begin{enumerate}
\item 
$ \  \Psi \cent \pi' \fixp t\sigma$, for all $\pi'\fixp^?t\in \probl'$, and $\Psi \cent \pi \fixp ([a]s)\sigma$.
 \item $ \  \Psi \cent t\sigma\faleq s\sigma$, for all $t\faleq^?s\in \probl'$.
  \end{enumerate}
Since 
$\Psi \cent \pi \fixp ([a]s)\sigma$ 
and $([a]s)\sigma= [a]s\sigma$, it follows that 
$\Psi \cent \pi \fixp [a](s\sigma)$.
From  inversion and rule $\rulefont{\fixp abs}$ (see Figure~\ref{fig.rules.fixp.gen}), this implies that there exists a proof for 
$\Psi, \overline{\swap{c_1}{c_2}\fixp \var{s\sigma}} \cent \pi \fixp \swap{a}{c_1}.s\sigma$.

Notice that we can always choose $c_1$ and $c_2$ such that $\supp{\swap{c_1}{c_2}}\cap \supp{s\sigma}=\emptyset$, from Lemma~\ref{lem:supp-fixp-faleq}, it follows that $\Psi \cent \swap{c_1}{c_2}\fixp s\sigma$. Since $\Psi, \overline{\swap{c_1}{c_2}\fixp \var{s\sigma}} \cent\pi \fixp \swap{a}{c_1}.s\sigma $, it follows that  $\Psi\cent\pi \fixp \swap{a}{c_1}.s\sigma $, by Proposition~\ref{prop:weak}.
The other inclusion is similar.

The cases corresponding to the other non-instantiating simplification rules  are similar to the above and omitted.
\qedhere
 \end{enumerate}
\end{proof}
\noindent
We now show that the result of the algorithm is a principal solution.

\begin{thm}
\label{thm.sol.is.principal}
Let $\probl$ be a unification problem, and suppose $\sol{\probl} = \pair{\Phi}{\sigma}$.
Then:
\begin{enumerate}
\item
$\pair{\Phi}{\sigma}\in\mathcal{U}(\probl)$, and
\item
$\pair{\Phi}{\sigma}\leq\pair{\Phi'}{\sigma'}$ for all other $\pair{\Phi'}{\sigma'}\in\mathcal{U}(\probl)$.
That is, the solution is also a \emph{least} or \emph{principal} solution.
\end{enumerate}
\end{thm}

\begin{proof}
We work by induction on the length of the reduction $\probl \Longrightarrow^{*} \sol{\probl}$.
\begin{itemize}
\item
Suppose $\probl$ is in normal form. Then the result is trivial since: 
\begin{enumerate}
\item  $\sigma=\Id$, $\Phi\cent\probl\Id$ and $\Id$ is idempotent;
\item For any other $\pair{\Phi'}{\sigma'}\in\mathcal{U}(\probl)$, $\sigma'$ is such that $\Phi'\cent \Phi\sigma$ and $\Phi'\cent X\Id\sigma\faleq X\sigma$ for all $X$.
\end{enumerate}
\item
Suppose $\probl\Longrightarrow \probl'$ by some non-instantiating simplification.  Then 
using Lemma~\ref{lem.preservation.of.solutions},
we know that $\mathcal{U}(\probl)=\mathcal{U}(\probl')$.  
Both parts of the result follow by induction.
\item
Suppose $\probl\Longrightarrow^{\theta}\probl'\theta$ by an instantiating rule.  Assume $\probl= \probl' \cup\{\pi\act X\faleq t\}$ where $\theta=[X\mapsto \pi^{-1}\act t]$ and $X\not\in \var{t}$ (the case for the other instantiating rule is similar).  Suppose $\sol{\probl'\theta}=\pair{\Phi}{\sigma}$, so that by construction $\sol{\probl}=\pair{\Phi}{\theta\circ\sigma}$. 
\begin{enumerate}
\item 
 It is easy to see that $\theta\circ\sigma$ is idempotent and by the first part of the inductive hypothesis $\Phi\cent \probl'\theta\sigma$, that is, $\pair{\Phi}{\theta\circ\sigma}\in\mathcal{U}(\probl)$.  
\item
Suppose $\pair{\Phi'}{\sigma'}\in\mathcal{U}(\probl)$.  
Then $\Phi'\cent X\sigma'\faleq \pi^{-1}\act t\sigma'$ by Equivariance. 
Hence, $\pair{\Phi'}{\theta\circ\sigma''}\in\mathcal{U}(\probl)$ where $\sigma''$ acts just like $\sigma'$ only it maps $X$ to $X$, $\theta=[X\mapsto \pi^{-1}\act t]$, and $\sigma'=\theta\circ\sigma''$.  Note that  $\pair{\Phi'}{\sigma''}\in\mathcal{U}(\probl'\theta)$ and by inductive hypothesis $\pair{\Phi}{\sigma}\leq \pair{\Phi'}{\sigma''}$.  It follows that $\pair{\Phi}{\theta\circ\sigma}\leq  \pair{\Phi'}{\theta\circ\sigma''}$. 
\qedhere
\end{enumerate}
\end{itemize}
\end{proof}

Matching and unification are closely related notions. While unification is the basis of logic programming languages, matching is at the heart of rewriting and functional programming. A matching algorithm can be easily derived from a unification algorithm. First we recall the definition of matching.

\begin{defi}
A {\em matching problem} is a particular kind of unification problem $\probl$ in which 
the variables in right-hand sides of equality constraints are disjoint from the variables in left-hand sides. A solution $\pair{\Phi}{\sigma}$ for a matching problem $\probl$ satisfies $\Phi \vdash s \faleq t\sigma$ and $\Phi \vdash \pi \fixp t\sigma$, for each $s\faleq^? t, \pi\fixp^? t \in \probl$ (i.e., in a matching problem, only the variables in right-hand sides of terms can be instantiated.
\end{defi}

A nominal matching algorithm can be obtained from the unification algorithm specified in Table~\ref{table:simpl.rules} by restricting the instantiation rules, so that only variables which were in left-hand sides of equality constraints in the initial problem can be instantiated.

\begin{obs}
Theorem~\ref{th:fix.alpha} guarantees the equivalence between $\aleq$ and $\faleq$, therefore, we can associate the unification algorithm proposed, with the standard nominal unification algorithm proposed in~\cite{Urban2004}. The problem $\probl$ introduced in Example~\ref{ex:unif.fix}, is equivalent to the nominal unification problem $\mathcal{P}= \{[a]\tf{f}(X,a)\aleq^? [b]\tf{f}(\swap{b}{c}\cdot W, \swap{a}{c}\cdot Y\}$, and using the standard simplification rules~\cite{Urban2004}:
\begin{equation}
\begin{aligned}
\mathcal{P}\stackrel{*}{\Longrightarrow}\stackrel{[Y\mapsto b]}{\Longrightarrow}& \stackrel{*}{\Longrightarrow}\mathcal{P}'= \{X\aleq^? \swap{a}{b}\cdot (\swap{b}{c}\cdot W),  a\#^? W\}\\
&\stackrel{[X\mapsto \swap{a}{b}\circ \swap{b}{c}\cdot W]}{\Longrightarrow} \{ a\#^? W\} = \mathcal{P}'
\end{aligned}
\end{equation}
The pair $\langle\mathcal{P}\rangle_{\tt sol}= \pair{\{a\#W\}}{\delta}$, where $\delta = \{Y\mapsto b, X\mapsto \swap{a}{b}\circ \swap{b}{c}\cdot W\}$ is a solution for $\mathcal{P }$. Using the translation $[\_]_\fixp$, we obtain $[\langle\mathcal{P}\rangle_{\tt sol}]_\fixp = \pair{\{[a\#W]_\fixp\}}{\delta}= \pair{\swap{a}{c_a}\fixp W}{\delta}$, where $c_a$ is a 
new name,
which is equivalent to $\pair{\swap{a}{c_a}\fixp W, \swap{c_a}{c_1}\fixp W}{\delta}$, for $c_a$ and $c_1$ not occurring anywhere in $\mathcal{P}$. 
Therefore, $[\langle\mathcal{P}_{\tt sol}\rangle]_\fixp$ is a solution for $\probl= \{[a]\tf{f}(X,a)\faleq^? [b]\tf{f}(\swap{b}{c}\cdot W, \swap{a}{c}\cdot Y\}$. Similarly, from the solution $\pair{\Psi}{\sigma}$ proposed in Example~\ref{ex:sol.unif}, we obtain 
$\pair{[\Psi]_\#}{\sigma}=\pair{a\#W,c_1\#W,c_2\#W}{\sigma}$, 
which is a solution for $\mathcal{P}$.
\end{obs}

In the theorem below $\probl_\fixp$ denotes a unification problem w.r.t. $\faleq$ and $\fixp$, and $\mathcal{P}_\#$ denotes a unification problem w.r.t. $\aleq$ and $\#$.
\begin{thm}\label{th:sol.unif.fix.fresh}
Let $\probl_\fixp$ and $\mathcal{P}_\#$ be unification problems such that $[\probl_\fixp]^\#= \mathcal{P}_\#$ and $\pair{\Psi}{\sigma} \in \mathcal{U}(\probl_\fixp)$ and $\pair{\Delta}{\delta}\in \mathcal{U}(\mathcal{P}_\#)$ be solutions for $\probl_\fixp$ and $\mathcal{P}_\#$, respectively. Then
\begin{enumerate}
\item $\pair{[\Psi]_\#}{\sigma} \in \mathcal{U}(\mathcal{P}_\#)$.
\item $\pair{[\Delta]_\fixp}{\delta} \in \mathcal{U}(\probl_\fixp)$.
\end{enumerate}
\end{thm}

\begin{proof}
Consequence of Theorems~\ref{lem:fresh.fixp} and \ref{th:aleq.to.faleq} and the fact that substitution preserves $\#$, $\aleq$, $\fixp$ and $\faleq$ judgements (see Proposition~\ref{prop:weak}).
\end{proof}

It is easy to see that the algorithm in Table~\ref{table:simpl.rules} is exponential (as the one given in~\cite{Urban2004}): even if we restrict the problems to first-order unification problems (without atoms), the simplification of the following problem requires a number of steps which is exponential with respect to the size of the problem.
$$ h(f(X_0,X_0),..., f(X_{n-1},X_{n-1})) \faleq^?  h(X_1, ..., X_n) $$ 
Comparing this algorithm with the one based on freshness constraints, one notices that there is a correspondence between the simplification rules in both approaches (the term syntax is the same in both cases). Moreover, the  simplification rules for fixed-point constraints work exactly like the simplification rules for freshness constraints (fixed-point rules have linear complexity, same as the simplification rules for freshness constraints in the freshness-based approach). 
The techniques designed to improve the efficiency of the freshness-based nominal unification algorithm (see, e.g.,  \cite{Calves2010a, Levy2010}) can equally  apply to the fixed-point based algorithm. To avoid the exponential complexity, terms should be represented via graphs, and permutations should be applied in a lazy way (see~\cite{Calves2010a, Levy2010} for details). 
We leave the development of efficient implementations to future work and in the rest of the paper we focus on unification modulo equational theories.

\section{Nominal alpha-equivalence modulo equational theories} 
\label{sec:alpha-equivalence-modulo}
In this section an extension of $\alpha$-equality to take into account equational theories will be proposed. 
It is well-known that first-order unification modulo an arbitrary equational theory \theory{E} is undecidable, therefore, it is expected that nominal unification modulo  \theory{E} (nominal \theory{E}-unification, for short) inherits this undecidability property. In this work, we propose an approach to deal with particular classes of equational theories, such as, associativity (\theory{A}), commutativity (\theory{C}) and associativity-commutativity (\theory{AC}), using fixed-point constraints.

In the case in which \theory{E} = \theory{A}, \theory{C}, \theory{AC}, or a combination of these theories, an algorithm to check \theory{E}-$\alpha$-equality  via freshness constraints was  proposed in~\cite{Ayala-Rincon2016, ACFONantes_jv19}, where  correctness results were formally verified using the Coq proof assistant, and an implementation in Ocaml was given.  Unification was considered only for  \theory{C} theories, for which it was shown that in general there is no finitary representation of the set of solutions if solutions are represented by freshness contexts and substitutions. 

We argue that the approach of fixed-point constraints is convenient when dealing with unification problems in equational theories that involve some notion of permutation of elements (such as commutativity), as it was shown in the previous version of this paper~\cite{DBLP:conf/rta/Ayala-RinconFN18}.

\subsection{Alpha-equivalence modulo \theory{E} via permutation fixed points}
In this section the relations  $\faleq$ and $\fixp$ will be extended to  $\ealeq{E}$ and $\efixp{E}$, where $\theory{E}$ is some equational theory. Inference rules 
will be parameterised by $\theory{E}$ so they can be reused when other theories are exploited. 

In this work  we have dedicated rules for \theory{A},  \theory{C} and \theory{AC}, and their application depends on whether the signature $\Sigma$ has function symbols satisfying these theories. Whenever we want to restrict the results to a particular theory, we will explicitly indicate the \theory{E}.

Similarly to the case of pure/syntactic $\alpha$-equality  we define notions of \theory{E}-fixed-point and \theory{E}-$\alpha$-equality constraints, as well as \theory{E}-fixed-point contexts and \theory{E}-judgments as expected.

\begin{defi}[\theory{E}-constraints and $\theory{E}$-fixed-point contexts]
\begin{itemize}
\item  An \theory{E}-{\em fixed-point constraint} is a pair of the form $\pi \fixp_{\theory{E}} t$, of a permutation $\pi$ and a term $t$. An \theory{E}-$\alpha$-{\em equality constraint} (for short, \theory{E}-equality constraint or just equality constraint) is a pair of the form $s\ealeq{E} t$, for nominal terms $s$ and $t$. 
\item We call a fixed-point constraint of the form $\pi\efixp{E} X$ a \emph{primitive \theory{E}-fixed-point constraint} and a finite set of such constraints is called  an \theory{E}-{\em fixed-point context}. 
 $\Upsilon, \Psi, \ldots$ range over  contexts. 
\end{itemize}
 \end{defi}
 
 Intuitively,  $s\ealeq{E} t$ will mean that $s$ and $t$ are $\alpha$-equivalent modulo the equational theory \theory{E}, and $\pi\fixp_\theory{E} t$ will mean that the permutation $\pi$ has no effect on the equivalence class of the term $t$ modulo \theory{E}.
 For instance, $\swap{a}{c}\efixp{C}+(a,c)$, assuming + is commutative, but not $\swap{a}{c}\efixp{C}\tf{f}(a,c)$, if $\tf{f} $ is not a commutative symbol.

 Below we assume that commutative symbols are always applied to pairs (although the grammar of nominal terms permits application of function symbols to tuples, we assume a syntactic check is carried out in the case of \theory{C}-symbols).

 \begin{defi}
An \theory{E}-\emph{fixed-point  judgement} is a tuple $\Upsilon \cent \new \overline{c}. \pi \efixp{E} t$ of a fixed-point  context and a fixed-point  constraint, possibly  with some newly quantified atoms, whereas an \theory{E}-$\alpha$-{\em equality judgement} is a tuple $\Upsilon\cent \new \overline{c}.  s \ealeq{E}t$
of a fixed-point context and an $\theory{E}$-equality constraint, also   with  some newly quantified atoms $\overline{c}$.

$\theory{E}$-judgements are derived using the rules in Figures \ref{fig.rules.efixp} and \ref{fig.rules.eequal}. 
 \end{defi}
 
 Notice that unlike the syntactical case, where the deduction rules for $ \fixp$ do not use $\faleq$, here
the rules for $\efixp{E}$ and $\ealeq{E}$ are mutually recursive when the theory $\theory{E}$ involves function symbols with commutative and associative-commutativity properties (see rules $\rulefont{\efixp{E}\tf{f}^\theory{C}}$ and $\rulefont{\efixp{E}\tf{f}^\theory{AC}}$.  
We assume that the terms are flattened w.r.t. associative and  associative-commutative function symbols.

Despite the fact that the rules are mutually recursive, the relations are well-defined since the recursion is well-founded. To see this, one can use a measure that interprets each equality judgement  by the pair consisting of the maximum of the sizes of terms in the equality constraint and the symbol $\ealeq{E}$, and each fixed-point judgement by the pair consisting of the
size of the term in the fixed-point constraint and the fixed-point symbol. For example, $\Upsilon \cent \new c. s \ealeq{E} t$ is interpreted by $\pair{max(|s|,|t|)}{\ealeq{E}}$
and $\Upsilon \cent \new c. (a\ b) \efixp{E} t$ 
by $\pair{|t|}{\efixp{E}}$. 
We compare these pairs lexicographically, using the
ordering on natural numbers to compare the 
sizes of terms, and the ordering  $\efixp{E} > \ealeq{E}$ for the second component of pairs. We can then check that in all the rules in Figures \ref{fig.rules.efixp} and \ref{fig.rules.eequal}  the premises are strictly smaller than the conclusion according to this ordering, therefore the rules provide an inductive definition of the set of derivable judgements.

Rules $\rulefont{\efixp{E} a}, \rulefont{\efixp{E} var}, \rulefont{\efixp{E} abs} $ and $\rulefont{\efixp{E} tuple}$ behave exactly as the corresponding rules in Figure~\ref{fig.rules.fixp} 
(where $\theory{E}= \emptyset$), i.e., the theory \theory{E} has no effect on the fixed-point constraint. 

Rule $\rulefont{\efixp{E}\tf{f}}$ is used for associative and uninterpreted function symbols. In the case of commutative or associative-commutative function symbols the rules  $\rulefont{\efixp{E} f^\theory{C}}$ and  $\rulefont{\efixp{E} f^\theory{AC}}$  are used. The goal is to ensure the analogous of Theorem~\ref{th:fix.alpha} for the relations $\efixp{E}$ and $\ealeq{E}$ that we wish to obtain: $\Upsilon \cent \new \overline{c} . \pi \efixp{E} t$ iff $\Upsilon \cent \overline{c} . \pi \cdot t \ealeq{E} t$ (Theorem~\ref{th:cfix.alpha}). This is illustrated in the example below.


 \begin{figure*}[ht]
 \hrule
$$
\begin{array}{@{\hspace{1.2mm}}l@{\hspace{1mm}}r@{\hspace{1.2mm}}}
\begin{prooftree}
\pi(a) = a
\justifies \Upsilon\cent \new \overline{c}.\pi\efixp{E} a
\using \rulefont{\efixp{E} a}
\end{prooftree}
&
\begin{prooftree}
\supp{\pi^{\pi'^{-1}}}\setminus \{\overline{c}\}\subseteq \supp{\perm{\Upsilon|_X}} 
\justifies \Upsilon\cent \new \overline{c}.\pi\efixp{E}\pi'\cdot X
\using \rulefont{\efixp{E} var}
\end{prooftree}
\\[4ex]
\begin{prooftree}
\Upsilon\cent \new \overline{c}.\pi\efixp{E} t
\justifies 
\Upsilon\cent \new \overline{c}.\pi\efixp{E} \tf{f}^\theory{E} t
\using \rulefont{\efixp{E}\tf{f}}, \mbox{if }\tf{f}^\theory{E}= \tf{f}, \tf{f}^\theory{A} 
\end{prooftree}
&
\begin{prooftree}
\Upsilon\cent \new \overline{c},c_1.\pi\efixp{E} \swap{a}{c_1} \act  t
\justifies 
\Upsilon\cent \new \overline{c}.\pi \efixp{E}  [a]t
\using \rulefont{\efixp{E} abs} 
\end{prooftree}
\\[4ex]
\end{array}
$$
$$
\begin{array}{c}
 \begin{prooftree}
\Upsilon\cent \new \overline{c}.\pi\efixp{E} t_1  \quad \ldots \quad \Upsilon\cent \new \overline{c}.\pi\efixp{C} t_n
\justifies
\Upsilon\cent \new \overline{c}.\pi\efixp{E}  (t_1,\ldots, t_n)
\using \rulefont{\efixp{E} tuple}
\end{prooftree}
\\[4ex]
 \begin{prooftree}
 \Upsilon\cent \new \overline{c}. \pi \cdot (\tf{f}^\theory{C}(t_0, t_1))\ealeq{E} \tf{f}^\theory{C}(t_0, t_1)
 \justifies
 \Upsilon\cent \new \overline{c}.\pi\efixp{E} \tf{f}^\theory{C} (t_0,t_1)
\using\rulefont{\efixp{E} \tf{f}^\theory{C}}
 \end{prooftree}
 \\[4ex]
 \begin{prooftree}
\Upsilon\cent \new \overline{c}. \pi  \cdot ( \tf{f}^\theory{AC} (t_0,t_1, \ldots, t_n)) \ealeq{E}  \tf{f}^\theory{AC} (t_0,t_1, \ldots, t_n)
\justifies
 \Upsilon\cent \new \overline{c}.\pi\efixp{E} \tf{f}^\theory{AC} (t_0,t_1, \ldots, t_n)
\using\rulefont{\efixp{E} \tf{f}^\theory{AC}}
 \end{prooftree}
\end{array}
$$
\hrule
\caption{Fixed-point rules modulo \theory{A}, \theory{C}, \theory{AC}. }
\label{fig.rules.efixp}
\end{figure*}

\begin{figure*}[ht]
\small
\hrule
$$
\begin{array}{l@{\hspace{-18mm}}r}
\begin{prooftree}
\justifies \Upsilon \cent \new \overline{c}. a \ealeq{E} a
\using \rulefont{\ealeq{E} a}
\end{prooftree}
&
\begin{prooftree}
 \supp{(\pi')^{-1}\circ \pi}\setminus \{\overline{c}\}\subseteq \supp{\perm{\Upsilon|_X}}
\justifies 
\Upsilon\cent  \new \overline{c}. \pi \cdot X \ealeq{E} \pi'\cdot X
\using \rulefont{\ealeq{E} var} 
\end{prooftree}
\\[4ex]
\begin{prooftree}
\Upsilon\cent\new \overline{c}. s\ealeq{E} \swap{a}{b} \cdot  t \quad \Upsilon\cent  \new \overline{c},c_1.\swap{a}{c_1} \efixp{E} t
\justifies
\Upsilon\cent \new \overline{c}. [a] s \ealeq{E} [b] t
\using \rulefont{\ealeq{E}{ab}}
\end{prooftree}
&
\begin{prooftree}
\Upsilon\cent  \new \overline{c}. t\ealeq{E} t' 
\justifies 
\Upsilon\cent \new \overline{c}. [a]t \ealeq{E} [a]t'
\using \rulefont{\ealeq{E}[a]} 
\end{prooftree}
\end{array}
$$
$$
\begin{array}{c}
\begin{prooftree}
\Upsilon\cent\new \overline{c}.  t\ealeq{E} t'
\justifies 
\Upsilon\cent \new \overline{c}. \tf{f}^\theory{E} t \ealeq{E} \tf{f}^\theory{E} t'
\using \rulefont{\ealeq{E}\tf{f}},
\begin{minipage}{4.0cm}
$\mbox{if } \tf{f}^\theory{E}= \tf{f}, \tf{f}^\theory{A} $
\mbox{ or }\\
$t,t' \mbox{ are not pairs}$
\end{minipage}
\end{prooftree}
\\[4ex]
\begin{prooftree}
\Upsilon\cent \new \overline{c}.  t_1\ealeq{E} t_1' \quad \ldots \quad\Upsilon \cent \new \overline{c}. t_n\ealeq{E} t_n'
\justifies 
\Upsilon\cent  \new \overline{c}. (\widetilde{t})_{1..n} \ealeq{E} (\widetilde{t'})_{1..n}
\using \rulefont{\ealeq{E} t}   
\end{prooftree}
\\[4ex]
    \begin{prooftree}
\Upsilon \cent \new \overline{c}. s_0\ealeq{E} t_i \quad \Upsilon \cent \new \overline{c}. s_1 \ealeq{E} t_{(i+1)\!\!\!\mod 2}
\justifies
\Upsilon \cent \new \overline{c}.\tf{f}^\theory{C}( s_0, s_1 )\ealeq{E} \tf{f}^\theory{C} (t_0, t_1) 
\using \rulefont{ \ealeq{E} \tf{f}^\theory{C}}
\mbox{ where }  i\in\{0,1\} 
\end{prooftree}
\\[4ex]
    \begin{prooftree}
\Upsilon \cent \new \overline{c}. s_0\ealeq{E} t_i \quad \Upsilon \cent \new \overline{c}.\tf{f}^\theory{AC}(\widetilde{s})_{1..n} \ealeq{E} \tf{f}^\theory{AC}(\widetilde{t})_{0..n}^{-i}
\justifies
\Upsilon \cent \new \overline{c}.\tf{f}^\theory{AC}(\widetilde{s})_{0..n}\ealeq{E} \tf{f}^\theory{AC} (\widetilde{t})_{0..n}  
\using \rulefont{ \ealeq{E} \tf{f}^\theory{AC}} ~~~~~
\begin{minipage}{5.5cm}
\mbox{where } $i \in \{0,\ldots,n\}$,\\ 
$(\widetilde{t_0})_n^{-i}=
(t_0, \ldots, t_{i-1}, t_{i+1}, \ldots, t_n)$
\end{minipage}
\end{prooftree}
\end{array}
$$
\hrule
\caption{Rules for equality modulo  \theory{A}, \theory{C}, \theory{AC}}
\label{fig.rules.eequal}
\end{figure*}

\begin{exa} Let $+$ be a commutative function symbol and suppose that we want to decide whether $\emptyset \cent \swap{a}{b}\fixp_\theory{C} ((a+b) +c)$, which corresponds to $\emptyset \cent \swap{a}{b} \cdot ((a+b)+c) \caleq (a+b)+c$ since $\swap{a}{b}\cdot ((a+b)+c) = (b+a)+c \caleq (a+b)+c$.
In general,  $\Upsilon\cent \pi \fixp_\theory{C} t_0 + t_1$ means that the permutation $\pi$ fixes $t_0+t_1$ modulo \theory{C} (given the information in $\Upsilon$), that is, $\Upsilon \cent \pi \cdot (t_0 + t_1) \caleq t_0 +t_1$. By definition, the permutation $\pi$ distributes homomorphically over the operator $+$, therefore, we have $\Upsilon \cent \pi \cdot t_0 + \pi \cdot t_1 \caleq t_0 +t_1 \caleq t_1 +t_0$. Thus, two cases can be distinguished: $\Upsilon \cent\pi \cdot t_i \caleq t_i$ or $\Upsilon \cent\pi \cdot t_i \caleq t_{(i+1)mod\  2}$, for $i=0,1$ (see rule \rulefont{\caleq f^\theory{C}} in Figure~\ref{fig.rules.eequal}).

\end{exa}

Similarly, $\alpha$-equality rules (see Figure~\ref{fig.rules.equal}) have their equational counterpart in Figure~\ref{fig.rules.eequal}: rules for atoms, tuples and abstractions are not affected by the theory $\theory{E}$, whereas rules involving function symbols $\tf{f}^\theory{E}$ have to be analysed separately.
\begin{itemize}
    \item Associative symbols: 
    rule $\rulefont{\ealeq{E}\tf{f}^\theory{A}}$ assumes that terms are flattened with respect to nested occurrences of $\tf{f}^\theory{A}$.
    \item Commutative symbols: 
    rule $\rulefont{\ealeq{E} \tf{f}^\theory{C}}$ is used. 
    \item Associative-commutative symbols: 
    terms are assumed to be flattened with respect to the $\tf{f}^\theory{AC}$ function symbol, and rule $\rulefont{\ealeq{E} \tf{f}^\theory{AC}}$ is used.
\end{itemize}

\begin{exa}
Consider the signature $\Sigma_\theory{A}=\{\tf{f}^\theory{A}\}\cup \Sigma^\emptyset$, where $\Sigma^\emptyset$ is a set of uninterpreted function symbols. Rules $\rulefont{\efixp{E} \tf{f}^\theory{AC}}$, $\rulefont{\efixp{E} \tf{f}^\theory{C}}, \rulefont{\ealeq{E} \tf{f}^\theory{C}}$ and $ \rulefont{\ealeq{E} \tf{f}^\theory{AC}}$ will not be used in judgements involving $\Sigma_ \theory{A}$-terms,  therefore,  we can  replace $\theory{E}$ for $\theory{A}$ and obtain rules for $\ealeq{A}$ and $\fixp_\theory{A}$. For instance, consider $\tf{f}^{\theory{A}} ( [a]a, \tf{f}^{\theory{A}}(b,X ))$, which is represented in flattened form as
$\tf{f}^{\theory{A}} ( [a]a, b,X )$.

\bigskip
\hspace{-6.5mm}\begin{prooftree}
\begin{prooftree}
\small
\[
     \begin{prooftree}
       \justifies 
       \Upsilon \cent \new c_1. \swap{a}{d}\efixp{A} c_1 
           \using \rulefont{\efixp{A} a}
      \end{prooftree}
\justifies
\Upsilon \cent \swap{a}{d}\efixp{A} [a]a
\using \rulefont{\efixp{A} abs}
\]\hspace{-3mm}
\begin{prooftree} 
     \justifies 
     \Upsilon \cent \swap{a}{d} \efixp{A} b
                \using \rulefont{\efixp{A} a}
     \end{prooftree}\hspace{-3mm}
     \begin{prooftree} 
      \{a,d\}\!\subseteq\! \supp{\perm{\Upsilon|_X}}
     \justifies
     \Upsilon \cent \swap{a}{d} \efixp{A} X
                \using \rulefont{\efixp{A} a}\;\;\rlap{\mbox{\bf ?}} 
                \end{prooftree}
                \justifies
 \Upsilon \cent \swap{a}{d} \efixp{A}  ( [a]a, b,X )
 \using \rulefont{\efixp{A} \tf{tuple}}
 \end{prooftree}
 \justifies 
\Upsilon \cent \swap{a}{d} \efixp{A} \tf{f}^{\theory{A}} ( [a]a, b,X )
\using \rulefont{\efixp{A} \tf{f}^\theory{A}}
\end{prooftree}
\normalsize
\bigskip

 The conclusion of this derivation depends on the support of the permutations in $\Upsilon$, this is illustrated with the question mark `{\bf ?}' in the rightmost leaf in the derivation. For example, if, on one hand, $\Upsilon = \emptyset$, then this derivation fails and we cannot conclude that $\swap{a}{d}$ fixes  $\tf{f}^{\theory{A}}([a]a, b,X)$; if, on the other hand, $\Upsilon = \{\swap{a}{b}\fixp X, \swap{d}{e}\fixp X\}\cup \Upsilon'$, then we could conclude the opposite.
\end{exa}

\begin{exa}
Consider the signature $\Sigma_{\{\theory{C,AC}\}}=\{\tf{\oplus}^\theory{C}, \tf{or}^\theory{AC} \}\cup\{\forall^\emptyset, \tf{g}^\emptyset\}\cup \Sigma^\emptyset$, where $\forall$ and $\tf{g}$ are unary uninterpreted function symbols. To improve readability, we will omit the superscripts of the function symbols in the rest of this example. Since we  only have commutative and associative-commutative symbols, we will replace $\theory{E}$ in the rules by $\{\theory{C,AC}\}$, therefore, obtaining rules for $\ealeq{\{C,AC\}}$ and $\efixp{\{C,AC\}}$. By applying the rules one can verify that
\begin{itemize}
    \item $\emptyset \not \cent \swap{a}{b}\swap{b}{c}\efixp{\{C,AC\}} (\tf{g}(a) \oplus \tf{g}(b)) \oplus \tf{g}(c)$, this is due to the fact that $\oplus$ is commutative but not associative, and the permutation $\swap{a}{b}\swap{b}{c}$ swaps the atom $a$ which is an argument of the inner $\oplus$ with the atom $c$ which is an argument of the outer 
    $\oplus$;
    \item $\emptyset \cent \swap{a}{b}\swap{b}{c}\efixp{\{C,AC\}} \tf{or}(\tf{or}(\tf{g}(a),\tf{g}(b)), \tf{g}(c))$;
    \item  $\emptyset \cent \swap{a}{b}\efixp{\{C,AC\}}\forall[a]\tf{or}(\tf{or}((a\oplus b),( b \oplus c)), (a \oplus c))  $; 
    \item $\swap{a}{b}\efixp{\{C,AC\}} X \cent\forall[a]\tf{or}(\tf{or}((a \oplus X) , g(c)),g(a))\ealeq{\{C,AC\}} \forall[b]\tf{or}(\tf{or}(g(c), (\swap{a}{b} \cdot X \oplus b)),g(b))$.
\end{itemize}
\end{exa}

The theorem below extends Theorem~\ref{th:fix.alpha}, relating fixed-point constraints to fixed-point equalities,
for the case in which equational theories $\theory{A}, \theory{C}$ and $ \theory{AC}$ are involved.
\begin{thm}\label{th:cfix.alpha}
Let $\Upsilon, \pi$ and $t$ be an \theory{E}-fixed-point context, a permutation and a nominal term, respectively. $\Upsilon\cent \new \overline{c}. \pi\efixp{E} t $ iff $\Upsilon\cent \new \overline{c}. \pi \cdot t\ealeq{E} t$.
\end{thm}

\begin{proof}

     


The proof is  by induction on the derivation,  similar to  the proof of Theorem \ref{th:fix.alpha}, except for the use of rules dealing with function symbols modulo $\theory{E}$. The rule for  associative symbols is the same as for syntactic symbols; we consider the cases of commutative and associative-commutative symbols.

\begin{itemize}
    \item 
     
     Suppose that $\Upsilon \cent \new \overline{c}. \pi \efixp{E} \tf{f}^\theory{C} (t_0, t_1)$, therefore, rule $\rulefont{\efixp{E} \tf{f}^\theory{C}}$ was applied, and one gets $\Upsilon\vdash \new \overline{c}. \pi \cdot (\tf{f}^\theory{C} (t_0, t_1)) \ealeq{E} \tf{f}^\theory{C} (t_0, t_1) $, and the result follows trivially. The other direction is also trivial.
    \item 

     Suppose that $\Upsilon \cent \new \overline{c}. \pi \efixp{E} \tf{f}^\theory{AC} (t_0,\ldots,  t_n)$, therefore, rule $\rulefont{\efixp{E} \tf{f}^\theory{AC}}$ was applied, and one gets $\Upsilon\vdash \new \overline{c}.  \pi \cdot (\tf{f}^\theory{AC} (t_0,\ldots,  t_n)) \ealeq{E} \tf{f}^\theory{AC} (t_0,\ldots,  t_n) $, and the result follows trivially. The other direction is also trivial.
     \qedhere
     \end{itemize}
\end{proof}


\subsection{From freshness to \theory{E}-fixed-point constraints and back again}

In~\cite{Ayala-Rincon2018}  relations $\appAA$, $\appAC$, $\appAAC$ and their combination $\appAll$ were defined as extensions of $\aleq$ using the standard approach via freshness constraints (see the
rules in Figures~\ref{fig.rules.freshness} and \ref{fig.rules.equaf}) but using specific rules
for associative, commutative and associative-commutative symbols. We recall those rules in Figure~\ref{fig:alpha_C_equivalence}.

\begin{figure}[http]
\figureboxed{
\[
\begin{array}{c}
\begin{prooftree}
\nabla \vdash s \appAll t 
\justifies \nabla \vdash  \tf{f}^\theory{E}\,s \appAll  \tf{f}^\theory{E}\,t 
\using \rulefont{\appAll  app},  
\mbox{ if } \theory{E} \notin \{\theory{A,C,AC}\}
\mbox{ or both } s \mbox{ and } t \mbox{ are not pairs}\;
\end{prooftree}
\\[3ex]
\begin{prooftree}
\nabla \vdash s_0 \appAll t_0 \qquad \tf{f}^\theory{A}(\widetilde{s})_{1..n} \appAll \tf{f}^\theory{A} (\widetilde{t})_{1..n}
\justifies
\nabla \vdash \tf{f}^\theory{A}(\widetilde{s})_{0..n} \appAll \tf{f}^\theory{A}(\widetilde{t})_{0..n}
\using \rulefont{\appAll \theory{A}}
\end{prooftree}
\\[3ex]
\begin{prooftree}
\nabla \vdash s_0 \appAll t_i \quad \nabla \vdash s_1 \appAll  t_{(i+1)\,mod\, 2}  \quad i = 0, 1
\using \rulefont{\appAll  \theory{C}} 
\justifies
\nabla \vdash \tf{f}^\theory{C}(s_0, s_1) \appAll  \tf{f}^\theory{C}(t_0, t_1)
\end{prooftree}
\\[3ex]
\begin{prooftree}
\nabla \vdash s_0 \appAll  t_i \quad \nabla \vdash\tf{f}^\theory{AC} (\widetilde{s})_{1..n} \appAll  \tf{f}^{AC}(\widetilde{t})_{0..n}^{-i}
\using \rulefont{\appAll  \theory{AC}} 
\justifies
\nabla \vdash \tf{f}^\theory{AC}(\widetilde{s})_{0..n} \appAll \tf{f}^\theory{AC}(\widetilde{t})_{0..n}
\end{prooftree}
\end{array}
\]
}
\caption{Additional rules for equational $\alpha$-equivalence via freshness constraints} 
\label{fig:alpha_C_equivalence}
\end{figure}

  Using a generalisation of the functions $[\_ ]_\fixp$ and $[\_ ]_\#$ defined in Section~\ref{sec:fresh.to.fixp}, we can obtain results that extend 
  Theorem~\ref{lem:fresh.fixp} and Theorem~\ref{th:aleq.to.faleq} to the equational case. The functions $[\_ ]_\fixp^\theory{E}$ and $[\_]_\#^\theory{E}$ defined below  are the natural extension of the previous translation 
  functions.

  The mapping $[\_]_\fixp^\theory{E}$ associates each primitive freshness constraint in $\Delta$ with a primitive fixed-point constraint:
  $$ [a \# X ]_\fixp^\theory{E} = \swap{a}{c_a}\fixp_E X  \mbox{ where } c_a \mbox{ is a new name}$$
This mapping extends directly to contexts. We denote by $[\Delta]_\fixp^\theory{E}$ the image of $\Delta$ under $[\_]_\fixp^\theory{E}$.

The mapping $[\_]_\#^\theory{E}$ associates each primitive fixed-point constraint in $\Upsilon$ with a primitive freshness context:
$$[\pi \fixp_\theory{E} X ]_\#^\theory{E} = \supp{\pi}\# X. $$
We denote by $[\Upsilon]_\#^\theory{E}$ the union of the freshness contexts obtained by translating each constraint in $\Upsilon$ using $[\_]_\#^\theory{E}$.
 
 \begin{thm}
 \label{lem:fresh.fixp.E}
 \begin{enumerate}
     \item 
 $\Delta\cent a\# t \Leftrightarrow [\Delta]^\theory{E}\cent \new c.\swap{a}{c}\efixp{E} t$.
 

\item $\Upsilon\cent \new \overline{c}. \pi\fixp_\theory{E} t \Leftrightarrow [\Upsilon]_\#^\theory{E}\cent \overline{\supp{\pi}\setminus\{\overline{c}\} \# t}$.
\end{enumerate}
  \end{thm}

 \begin{proof} The proof follows the same lines of the proof of Theorem~\ref{lem:fresh.fixp}. We discuss only the proof of  part (1).
  The proof is by induction on rules of Figure~\ref{fig.rules.freshness} used in the derivation of $\Delta \cent a\# t$. We show only the cases corresponding to  function symbols.  

\noindent $(\Longrightarrow)$
\begin{itemize}
     \item Rule $\rulefont{\# f}$
     
     In this case $t=\tf{f}^\theory{E}t'$, for some theory $\theory{E}$. The analysis is based on the specific theory:
      \begin{enumerate}
          \item $\tf{f}^\theory{E}=\tf{f}^\theory{A}$
             
             This case is analogous to the case in which $\theory{E}=\emptyset$.

           \item $\tf{f}^\theory{E}=\tf{f}^\theory{C}$

          There is a proof $\Pi$ of the form
          
           \begin{prooftree}
                   \[\Pi 
          \justifies
          \Delta \cent a \#  (t_1,t_2)\]
          \justifies
          \Delta \cent a \# \tf{f}^\theory{C} (t_1,t_2)
          \using \rulefont{\# f} 
           \end{prooftree}
          
         \bigskip
        
        By induction hypothesis, there exists a proof $\Pi'$ of $[\Delta]^\theory{E}_\fixp\cent \new c_a. \swap{a}{c_a}\efixp{E}(t_1,t_2)$. Therefore,  there is a proof of the form 
        
         \begin{prooftree}
         \begin{prooftree}
        \Pi_1^{'}
         \justifies 
         [\Delta]^\theory{E}_\fixp\cent \new c_a. \swap{a}{c_a}\efixp{E}t_1
         \end{prooftree} 
         \quad 
         \begin{prooftree}
         \Pi_2^{'}
         \justifies 
         [\Delta]^\theory{E}_\fixp\cent \new c_a. \swap{a}{c_a}\efixp{E}t_2
         \end{prooftree}
          \justifies
          [\Delta]^\theory{E}_\fixp\cent \new c_a. \swap{a}{c_a}\efixp{E}(t_1,t_2)
          \using 
          \rulefont{\efixp{E} tuple}
         \end{prooftree}
         \bigskip

        From $[\Delta]^\theory{E}_\fixp\cent \new c_a. \swap{a}{c_a}\efixp{E} t_i$ one can derive that $[\Delta]^\theory{E}_\fixp\cent \new c_a. \swap{a}{c_a}\cdot t_i \ealeq{E}t_i$, by Theorem~\ref{th:cfix.alpha}, for ($i=1,2$). By applying rule $\rulefont{\ealeq{E} \tf{f}^\theory{C}}$, it follows that  $$[\Delta]^\theory{E}_\fixp\cent \new c_a. \swap{a}{c_a}\cdot (\tf{f}^\theory{C}(t_1, t_2)) \ealeq{E}\tf{f}^\theory{C}(t_1, t_2),$$
        
        and the result follows from Theorem~\ref{th:cfix.alpha}. 
        
         \item $\tf{f}^\theory{E}=\tf{f}^\theory{A}$ 
         
         The proof is analogous to the case above. 
         
 \end{enumerate}
\end{itemize}
 
 ($\Longleftarrow$) The interesting case is again for rule $\rulefont{\efixp{E}\tf{f}^\theory{C}}$.
 
 Suppose that $[\Delta]_\fixp^\theory{E} \cent \new c. \swap{a}{c} \efixp{E} t_1 \oplus t_2$. We want to prove that $\Delta \cent a \# t_1\oplus t_2$. From rule $\rulefont{\efixp{E}\tf{f}^\theory{C}}$ we can conclude 
 that
 
 \begin{itemize}
     \item either there exist proofs of  $[\Delta]_\fixp^\theory{E} \cent \new c. \swap{a}{c} \cdot t_1 \ealeq{E} t_1$ and $[\Delta]_\fixp^\theory{E} \cent \new c. \swap{a}{c} \cdot t_2 \ealeq{E} t_2$, and by Theorem~\ref{th:cfix.alpha} it follows that there exist proofs of  $[\Delta]_\fixp^\theory{E} \cent \new c. \swap{a}{c} \efixp{E} t_1 $ and $[\Delta]_\fixp^\theory{E} \cent \new c. \swap{a}{c} \efixp{E}t_2$, and the result follows by induction hypothesis.

     \item  or there exist proofs of  $[\Delta]_\fixp^\theory{E} \cent \new c. \swap{a}{c} \cdot t_1 \ealeq{E} t_2$ and $[\Delta]_\fixp^\theory{E} \cent \new c. \swap{a}{c} \cdot t_2 \ealeq{E} t_1$.
     
     Since these equalities hold for any new name $c$ (not occurring in $t_1$ or $t_2$), then $a$ cannot be in the support of $t_1$ and $t_2$ and therefore $a$ is fresh in $t_1$ and $t_2$.
     \qedhere
     \end{itemize}
 \end{proof}

We can now relate $\alpha$-equivalence modulo $\theory{E}$  via freshness constrains ($\appAE$) with its version via fixed-point constraints ($\ealeq{E}$).

 \begin{thm}\label{thm:aleq.to.appAc}
 \begin{enumerate}
 \item $\Upsilon\cent \new\overline{c}. \ s\ealeq{E} t \Rightarrow [\Upsilon]_\#^\theory{E}\cup \Delta \cent s\appAE t$, where $\Delta \cent \overline{c}\# \var{s, t} $.
\item $\Delta\cent s\appAE t \Rightarrow [\Delta]^\theory{E}_\fixp \cent s\ealeq{E}  t$.
 \end{enumerate}
 \end{thm}

\begin{proof}
The proof is very similar to the proof of Theorem~\ref{th:aleq.to.faleq}, except for the case of rules involving function symbols $\tf{f}^\theory{E}$, with $\theory{E}
\neq \emptyset$, where the reasoning is similar to the one in the proof of Theorem~\ref{lem:fresh.fixp.E}.
\end{proof}
%
%
 

\subsection{Solving nominal \theory{C}-unification problems via fixed-point constraints}

 In this section we propose an approach to nominal unification modulo commutativity via the notion of fixed-point constraints. 
 
 For example, assuming $+$ is commutative, i.e., $X+Y = Y+X$, a  problem of the form
\begin{equation}\label{eq:c.unif.prob}
+(\swap{a}{b}\cdot X,a)\faleq^? +(Y,X)
\end{equation}
can be solved by unifying $\swap{a}{b}\act X$ with $Y$ and $a$ with $X$, or
$\swap{a}{b}\act X$ with $X$ and $a$ with $Y$.

 In~\cite{Ayala-Rincon2018}, a simplification algorithm for solving nominal \theory{C}-unification was proposed. This algorithm  was based on the standard nominal unification algorithm~\cite{Urban2004} where $\alpha$-equivalence is defined w.r.t. the notion of freshness. Upon the input of a unification problem $\probl$, the algorithm outputs a finite family of triples of the form $\pair{\nabla, \sigma}{P}$, where $\nabla$ is a freshness context, $\sigma$ a substitution and $P$ is a set of fixed-point equations, which are solved using a separate procedure.
 
 In \cite{Ayala2017} it is proved that even a simple 
 unification problem such as $\swap{a}{b}\act X \aleq X$ (i.e., a problem consisting of just one fixed-point equation) could produce an infinite and independent set of solutions, whenever the signature contains commutative function symbols. For example, if $f$ is commutative, the following substitutions solve this equation: $\{X \mapsto a+b, X \mapsto f(a+b), X \mapsto [e](a+b,b+a),\ldots \}$. Therefore, it is not possible to obtain a finite and complete set of solutions  for every solvable unification problem if solutions are expressed using freshness constraints and substitutions.
 However, we remark that the problem $+(\swap{a}{b}\cdot X,a)\faleq^? (Y,X)$ mentioned above has in fact a finite number of most general solutions (indeed, two) if we solve it using fixed-point constraints. The most general unifiers are $\{X \mapsto a, Y \mapsto b\}$ and $\{Y \mapsto a, \swap{a}{b}\fixp X\}$. This observation led us to use fixed-point constraints instead of freshness constraints to express solutions.

 Similarly to Section~\ref{sec:unif}, below we define the notion of nominal \theory{C}-unification in terms of \theory{C}-fixed-point constraints, and provide a nominal \theory{C}-unification  unification algorithm specified by means of simplification rules. 
 
 In this section, as in Section~\ref{sec:unif}, we assume a generator of new names exists, and remove the new quantifier from the syntax of unification problems.

 \begin{defi}
 A \theory{C}-{\em unification problem} ${\probl}$ is a pair $\pair{\Phi}{P}$ where $P$ is a finite set of \theory{C}-equality constraints 
 $ s\ucaleq t$ 
 and $\Phi$ is a finite set of \theory{C}-fixed-point constraints 
 $\pi \efixp{C}^? t$.
 To ease the notation, we will denote $s\ucaleq t$ by $s\approx^? t$.
 \end{defi}

 \begin{defi}[Solutions of $\theory{C}$-unification problems]\label{def:c.probl}
 A solution for a \theory{C}-unification problem $\probl=\pair{\Phi}{P}$ is a pair $\pair{\Upsilon}{\sigma}$, where the following conditions are satisfied
  \begin{enumerate}
   
      \item $\Upsilon \cent  \pi \efixp{C}t\sigma$, if $  \pi \efixp{C}^? t \in \Phi$;
     \item $\Upsilon \cent  s\sigma \ealeq{C} t\sigma$, if $ s\approx^? t \in P$.
     \item $\Upsilon \cent  X\sigma \sigma \ealeq{C} X\sigma$.
  \end{enumerate}
  \end{defi}

  The set of solutions for a \theory{C}-unification problem $\probl$ is denoted as  $\mathcal{U}_\theory{C}(\probl)$.

 \begin{defi}[Most general solution and complete set of solutions]
   \leavevmode
  \begin{itemize}
      \item For $\pair{\Upsilon}{\sigma}$ and $\pair{\Psi}{\delta}$ in $\mathcal{U}_\theory{C}(\probl)$, we say that $\pair{\Upsilon}{\sigma}$ is \emph{ more general than}
      $\pair{\Psi}{\delta}$, denoted $\pair{\Upsilon}{\sigma} \preceq \pair{\Psi}{\delta}$, if there exists a substitution $\rho$ satisfying $\Psi \cent \sigma \rho \ealeq{C} \delta$ and $ \Psi \cent \Upsilon \rho $.
      \item  A subset $\mathcal{C}$ of $\mathcal{U}_\theory{C}(\probl)$ is a \emph{complete set of solutions} of $\probl$ if for all $\pair{\Psi}{\sigma}\in \mathcal{U}_\theory{C}(\probl)$, there exits a $\pair{\Upsilon}{\delta} \in \mathcal{C}$ such that $\pair{\Upsilon}{\delta}\preceq \pair{\Psi}{\sigma}$. We denote a complete set of solutions of the \theory{C}-unification problem $\probl$ as $\mathcal{C}(\probl)$.
  \end{itemize}
 \end{defi}

Table~\ref{table:c.simpl.rules} presents the simplification rules for \theory{C}-unification problems. They are derived from the deduction rules for judgements, as done for the syntactic case. The main difference is that now there are  two rules for the simplification of fixed-point constraints involving commutative symbols (rules $\rulefont{\efixp{C} \tf{f}^\theory{C}1}$ and $\rulefont{\efixp{C} \tf{f}^\theory{C}2}$) and two rules to deal with equality of terms rooted by a commutative symbol (rules \rulefont{\caleq \tf{f}^\theory{C}1} and \rulefont{\caleq \tf{f}^\theory{C}2}).

\begin{table}[ht]
 \hrule
 \small
 \[
 \begin{array}{llll}
(\efixp{C} \ at) &\probl\uplus \{\pi \efixp{C}^? a\} &\Longrightarrow& \probl, \mbox{ if }\ \pi(a) =a\\
(\efixp{C} \tf{f}^{\emptyset})& \probl\uplus \{\pi \efixp{C}^? \tf{f}t\} &\Longrightarrow &\probl\cup \{\pi \efixp{C}^? t\}, \tf{f}\mbox{ not } \theory{C}\\
 (\efixp{C} \tf{f}^\theory{C}1)& \probl\uplus \{\pi \efixp{C}^?  \tf{f}^\theory{C}(t_0,t_1)\} &\Longrightarrow &\probl\cup \{\pi\act t_0\approx^? t_0, \pi\act t_1\approx^? t_1\}\\
 (\efixp{C} \tf{f}^\theory{C}2)& \probl\uplus \{\pi \efixp{C}^? \tf{f}^\theory{C}(t_0,t_1)\} &\Longrightarrow &\probl\cup \{\pi\act t_0\approx^? t_1, \pi\act t_1\approx^? t_0\}\\
(\efixp{C}\  tuple)&\probl\uplus \{\pi \efixp{C}^?    (\widetilde{t})_{1..n}\}&\Longrightarrow &\probl\cup\{\pi \efixp{C}^?t_1, \ldots, \pi\efixp{C}^? t_n\}\\
(\efixp{C} abs)& \probl\uplus \{\pi \efixp{C}^? [a]t\} &\Longrightarrow &\probl\cup 
\{\pi \efixp{C}^? \swap{a}{c_1} \act t, \overline{\swap{c_1}{c_2} \efixp{C} \var{t}}\}\\
(\efixp{C} var)& \probl\uplus \{\pi \efixp{C}^? \pi'\cdot X\}&\Longrightarrow & \probl\cup \{\pi^{(\pi')^{-1}}\efixp{C}^? X\}, \mbox{ if } \pi' \neq Id
\\
(\caleq a)&\probl\uplus \{a\approx^? a\}&\Longrightarrow & {\probl}\\
 (\caleq \tf{f})&\probl\uplus \{\tf{f}t\approx^? \tf{f}t'\}&\Longrightarrow & \probl\cup \{t\approx^? t'\}, \ \tf{f}\mbox{ not } \theory{C}\\
 (\caleq \tf{f}^\theory{C}1)&\probl\uplus \{\tf{f}^\theory{C}(t_0,t_1)\approx^?\tf{f}^\theory{C}(s_0,s_1)\}&\Longrightarrow & \probl\cup \{t_0\approx^? s_0, t_1\approx^? s_1\}\\
 (\caleq \tf{f}^\theory{C}2)&\probl\uplus \{\tf{f}^\theory{C}(t_0,t_1)\approx^?\tf{f}^\theory{C}(s_0,s_1)\}&\Longrightarrow & \probl\cup \{t_0\approx^?s_1, t_1\approx^? s_0\}\\
(\caleq tuple)&\probl\uplus \{(\widetilde{t})_{1..n}\approx^? (\widetilde{t'})_{1..n}\} &\Longrightarrow& {\probl\cup \{t_1\approx^?  t'_1, \ldots, t_n\approx^? t'_n\}}\\
(\caleq abs1)&\probl\uplus \{[a]t\approx^? [a]t'\}&\Longrightarrow & {\probl\cup\{t\approx^? t'\}}\\
(\caleq abs2)&\probl\uplus \{[a]t\approx^? [b]s\} &\Longrightarrow
&\probl\cup \{t\approx^?\swap{a}{b}\act s, \swap{a}{c_1}\efixp{C}^?s, \\
&&&\hfill\overline{\swap{c_1}{c_2} \efixp{C} \var{s}} \}  \\
(\caleq var)&\probl\uplus \{\pi \cdot X \approx^? \pi'\act X\} &\Longrightarrow & \probl\ \cup \{(\pi')^{-1}\circ \pi \efixp{C}^? X\}\\
(\caleq inst1)& \probl\uplus \{\pi \cdot X \approx^? t\} &\hspace{-3.5mm}\stackrel{[X \mapsto\pi^{-1}.t]}\Longrightarrow & \probl\{X \mapsto\pi^{-1}.t\}, \mbox{ if } X\notin \var{t}\\ 
(\caleq inst2)& \probl\uplus \{ t \approx^? \pi \cdot X\} &\hspace{-3.5mm}\stackrel{[X \mapsto\pi^{-1}.t]}\Longrightarrow & \probl\{X \mapsto\pi^{-1}.t\}, \mbox{ if } X\notin \var{t}\\ 
 \end{array}
 \]
 \hrule
 \caption{Simplification Rules for \theory{C}-unification problems via $\theory{C}$-fixed-point constraints. 
 In rules $(\fixp_C abs)$ and $(\caleq abs2 )$, $c_1$ and $c_2$ are newly generated names.}
 \label{table:c.simpl.rules}
 \end{table}

   We write $\probl\Longrightarrow_\theory{C} \probl'$ when $\probl'$ is obtained from $\probl$ by applying a simplification rule from Table~\ref{table:c.simpl.rules} and we write $\stackrel{*}{\Longrightarrow}_\theory{C}$ for the reflexive and transitive closure of $\Longrightarrow_\theory{C}$. We omit the subindex when it is clear from the context.

 \begin{lem}[Termination of simplification for $\theory{C}$-unification problems]
 \label{lem:C-termination}
	There is no infinite chain of reductions $\Longrightarrow_\theory{C}$ starting from a \theory{C}-unification problem $ \probl$.
 \end{lem}
 
 \begin{proof}

 Termination of the simplification rules follows directly from the fact that the following measure of the size of $\probl$ is  strictly decreasing:
 $[\probl] = (n_1,M)$ where $n_1$ is the number of different variables 
 used in $\probl$,
 and $M$ is the multiset of \emph{heights} of equality constraints and non-primitive fixed-point constraints occurring in $\probl$.
 
Each simplification step either eliminates one variable (when an instantiation rule is used) and therefore decreases the first component of the interpretation, or leaves the first component unchanged but replaces a constraint with primitive ones and/or constraints where terms have smaller height.
 \end{proof}

 The simplification rules (Table~\ref{table:c.simpl.rules}) specify a \theory{C}-\emph{unification algorithm}:  we apply the simplification rules in a problem $\probl$  until we reach  normal forms. In the case of a term rooted by a commutative symbol, two rules can be applied, so a tree of derivations is built. The termination property (Lemma~\ref{lem:C-termination}) guarantees the tree is finite.

 For the leaves in the tree (i.e., normal forms), the notions of \emph{consistency, failure, correctness} can be defined as in   Section~\ref{sec:unif} (see Definition~\ref{def:successfulnf}). So, if a normal form contains equality constraints, or inconsistent fixed-point constraints of the form 
 $\pi \efixp{C}^? a$ such that $\pi(a) \neq a$ 
 then this normal form is a failure. Only leaves containing consistent fixed-point constraints produce solutions.

We now prove that the \theory{C}-unification algorithm is sound and complete. The proof is done in two stages, first we show that the non-instantiating rules preserve solutions if we consider all the branches of the derivation tree (here it is important to consider all the branches: due to the non-deterministic application of rules involving commutative operators, if we consider just one branch we may loose solutions). Then we show that the set of solutions computed from all the successful leaves is a complete set of solutions for the initial problem.

 \begin{lem}[Correctness of non-instantiating rules] 
\label{lem.cpreservation.of.solutions}
Let $\probl$ be a \theory{C}-unification problem and $n$ a natural number. Assume  $\probl\overset{{n}}{\Longrightarrow_\theory{C}}\probl'_i$ $(i \in I)$ are all the 
reduction sequences of length smaller than or equal to $n$ starting from $\probl$ that do not use
instantiating rules $(\caleq inst1)$ and $(\caleq inst2)$.  Then
\begin{enumerate}
 \item
 $\mathcal{U}_\theory{C}(\probl)= \bigcup_{i\in I}\mathcal{U}_\theory{C}(\probl'_i)$, and 
\item
if $\probl'_i$ contains  inconsistent reduced fixed-point constraints then $\mathcal{U}_\theory{C}(\probl'_i)=\emptyset$.
 \end{enumerate}
\end{lem}
\begin{proof}
For part (1), 
as for the proof of Lemma~\ref{lem.preservation.of.solutions}, we proceed by induction on $n$, but here we need to consider all the branches of length smaller than or equal to $n$ in the derivation tree to ensure completeness.

The interesting cases are for the rules involving $\theory{C}$ function symbols, all the other cases are very similar to the proof of Lemma~\ref{lem.preservation.of.solutions}.  
\begin{itemize}
\item Suppose that the last step of a simplification chain  has the form
 \begin{equation*}
 \begin{aligned}
\probl_{n-1}&=\probl'\uplus \{\pi \efixp{C}^? \tf{f}^\theory{C}(s_0,s_1)\} \Longrightarrow \probl'\cup \{\pi \cdot s_0 \ealeq{C}^? s_i,  \pi \cdot s_1 \ealeq{C} s_{(i+1)mod 2}\}= \probl_n^i.
\end{aligned}
\end{equation*}

In this case, the rule used is either $(\efixp{C} \tf{f}^\theory{C}1)$ or $(\efixp{C} \tf{f}^\theory{C}2)$. 

Assume $i = 0$ and $(\efixp{C} \tf{f}^\theory{C}1)$ was used (the case $i = 1$ is identical). 

There is another reduction sequence of the same length using $(\efixp{C} \tf{f}^\theory{C}2)$ and ending on $\probl_n^1$ (the same problem with $i = 1$).

Let $\pair{\Psi}{\sigma}\in \mathcal{U}(\probl_{n-1}$) be a solution for $\probl_{n-1}$:
\begin{enumerate}
 \item 
 $\Psi \cent \pi' \efixp{C} t\sigma$, for all $\pi'\fixp^?t\in \probl'$ and $\Psi \cent \pi \efixp{C} \tf{f}^\theory{C}(s_0,s_1)\sigma$.
 \item $ \Psi \cent t\sigma\ealeq{C} s\sigma$, for all $t\approx^?s\in \probl'$.
  \end{enumerate}

Since 
$\Psi \cent \pi \efixp{C} \tf{f}^\theory{C}(s_0,s_1)\sigma$
and $\tf{f}^\theory{C}(s_0,s_1)\sigma= \tf{f}^\theory{C}(s_0\sigma,s_1\sigma)$, it follows that 
$\Psi \cent \pi \efixp{C} \tf{f}^\theory{C}(s_0\sigma,s_1\sigma)$. From  $\rulefont{\efixp{E} \tf{f}^\theory{C}}$, one has  that either there exist a  proof for 
$\Psi \cent \pi \cdot s_0\sigma \ealeq{C}s_i\sigma$ 
and 
$\Psi \cent \pi \cdot s_1\sigma \ealeq{C}s_{(i+1)mod\ 2}\sigma$, 
for $i=0$ or for $i = 1$. Therefore, 
$\pair{\Psi}{\sigma}\in \mathcal{U}(\probl_{n}^0)$ or $\pair{\Psi}{\sigma}\in \mathcal{U}(\probl_{n}^1)$ and the result follows.

The other direction is similar:  if $\pair{\Psi}{\sigma}\in \mathcal{U}(\probl_{n}^i)$, then $\pair{\Psi}{\sigma}\in \mathcal{U}(\probl_{n-1})$.

\item Suppose the last step of the simplification chain has the form 
\begin{equation*}
 \begin{aligned}
 \probl_{n-1}&=\probl'\uplus \{ \tf{f}^\theory{C}(s_0,s_1)\ealeq{C} \tf{f}^\theory{C}(t_0,t_1) \} \Longrightarrow  \probl_n^i
\end{aligned}
\end{equation*}
where  $\probl_n^i=\probl'\cup 
\{ s_0\ealeq{C}t_i,  s_1\ealeq{C}t_{(i+1)mod\ 2}\}$.

As above, it follows that there two branches of this form, for $i = 0$ and $i=1$.

Let $\pair{\Psi}{\sigma}\in \mathcal{U}(\probl_{n-1}$) be a solution for $\probl_{n-1}$. In particular, $\Psi \cent 
\tf{f}^\theory{C}(s_0\sigma,s_1\sigma)\ealeq{C}\tf{f}^\theory{C}(t_0\sigma,t_1\sigma)$. By applying rule $\rulefont{\ealeq{E} \tf{f}^\theory{C}}$, one has that there exist proofs of $\Psi \cent 
s_0\sigma\ealeq{C}t_i\sigma$ and $\Psi \cent 
s_1\sigma\ealeq{C}t_{(i+1)mod\ 2}\sigma$, for $i=0$ or $i=1$. Therefore, 
$\pair{\Psi}{\sigma}\in \mathcal{U}(\probl_{n}^i)$ for $i=0$ or $i=1$, and the result follows.

Similarly, if $\pair{\Psi}{\sigma}\in \mathcal{U}(\probl_{n}^i)$, then $\pair{\Psi}{\sigma}\in \mathcal{U}(\probl_{n-1})$.
\end{itemize}

The second part of the lemma follows directly from the fact that inconsistent constraints are not derivable, therefore have no solutions.
\end{proof}

As in Section~\ref{sec:unif},  when $\probl$ is a successful leaf, $\sol{\probl}$ consists of the composition $\sigma$ of all substitutions applied through the simplification steps and the fixed point context obtained.

\begin{thm}[Soundness and Completeness]
 Let $\probl=\pair{\Upsilon}{P}$ be a \theory{C}-unification problem and let
 $\{\probl'_i \mid  \probl \overset{*}{\Longrightarrow_\theory{C}} \probl'_i ~and~  \probl'_i \mbox{ successful  normal form}\}$ be the set of all the successful normal forms of $\probl$ (i.e., leaves in the derivation tree without equality constraints or inconsistent fixed-point constraints).

\begin{enumerate}
\item
If $\pair{\Phi}{\sigma}\in \bigcup \sol{\probl'_i}$ then  $\pair{\Phi}{\sigma}\in\mathcal{U}(\probl)$, and
\item  If $\pair{\Phi}{\sigma}\in\mathcal{U}(\probl)$, there exists 
$\pair{\Phi'}{\sigma'}$ such that $\pair{\Phi'}{\sigma'}\in \bigcup \sol{\probl'_i}$ and  $\pair{\Phi'}{\sigma'}\preceq\pair{\Phi}{\sigma}$, 
that is, the set $\bigcup \sol{\probl'_i}$ is a complete set of solutions.
\end{enumerate}
\end{thm}
 
 \begin{proof}
 The proof is by induction on the length of a derivation $\probl {\Longrightarrow}^*\probl'_i,$ distinguishing cases according to the first rule used.
 By Lemma \ref{lem.cpreservation.of.solutions}, it is sufficient to check generality of solutions after each application of instantiation rules.
If the first step uses a non-instantiating rule, then the previous lemma, together with the induction hypothesis, ensures that the set of  solutions of $\probl$ is exactly the set of solutions of its children. If the first step is instantiating, we proceed as in the proof of Theorem~\ref{thm.sol.is.principal}.
 \end{proof}
 
 \begin{obs}
 Using the approach to nominal $\alpha$-equivalence via freshness, a nominal  \theory{C}-unification algorithm was presented in~\cite{Ayala-Rincon2018,Ayala2017}, which outputs solutions represented as triples   $\pair{\nabla, \sigma}{P}$ consisting
of a freshness context $\nabla$, a substitution $\sigma$ and a set  $P$  of fixed-point equations
of the form $\pi \cdot X\appAC^? X$.

As with standard nominal unification, one can use the functions $[\_]_\#$ and $[\_]_\fixp$ to translate solutions $\pair{\nabla, \sigma}{P}$ of nominal \theory{C}-unification problems with freshness constraints as solutions $\pair{[\nabla]_{\fixp}\cup \{P_{\efixp{C}}\}}{\sigma}$ of nominal \theory{C}-unification problems via \theory{C}-fixed-point constraints,
where $P_{\efixp{C}}=\{\pi \efixp{C} X\ |\ \pi \cdot X\appAC^? X\in P \}$.
\end{obs}

A set of simplification rules generalising the  \theory{C}-unification algorithm to take into account  \theory{A} and \theory{AC} symbols was proposed in Ribeiro's thesis \cite{WashThese2019} following the freshness constraint approach.  The main difficulty is in the rules to deal with the \theory{AC} symbols and with treatment of fixed-point equations of the form  $\pi\cdot X \approx^?_{\{\alpha,\theory{AC}\}} X$.  Analytical proofs of soundness and completeness of such rules were given, and a   formalisation in Coq was developed for the \theory{C}-unification algorithm presented in \cite{Ayala-Rincon2018}.  In future work we will consider and relate \theory{AC}-nominal unification with freshness and fixed-point constraints. 

Regarding the complexity of the \theory{C}-unification algorithm based on fixed-point constraints, we observe that in the syntactic case (i.e., without $\alpha$-equivalence rules), the \theory{C}-unification problem is NP-complete so it is expected that the algorithm will be exponential (Chapter 10 on Equational Unification in \cite{BaNi98} surveys in detail the cases of \theory{C} and \theory{AC} unification). Comparing the nominal unification modulo \theory{C} based on freshness and on fixed-point constraints, we can again notice that there is a one-to-one correspondence in the simplification rules, and thus the algorithms have the same behaviour during the simplification phase. The main difference is that using the freshness approach, a second algorithm is needed to solve fixed-point equations (generating an infinite number of solutions in general), which is avoided with the fixed-point approach. 

\section{Conclusions and Future Work}
\label{sec:conclusions}
The notion of fixed-point constraint allowed us to obtain a finite representation of solutions for nominal \theory{C}-unification problems. 
This brings a novel alternative to standard nominal unification approaches in which just the algebra of atom permutations and the logic of freshness constraints are used to implement equational reasoning (e.g., \cite{Aoto2016a,Calves2013,Calves2008a,Cheney2010,Fernandez2004}), and in particular to their extensions modulo commutativity, for which only infinite representations were  possible in the standard approach.    We have shown that with the new proposed approach the development of an algorithm to solve nominal equational problems modulo \theory{C} is simpler,  avoiding, thanks to  the use of fixed-point constraints, the development of procedures for the generation of infinite independent sets of solutions.

In future work we plan to extend this approach to matching and unification modulo different equational theories as well as to the treatment of equational problems in nominal rewriting modulo.  Future study will also address handling the case of Mal'cev permutative theories, which include $n$-ary functions with permutative arguments \cite{HComon93}, as well as the more general and complex case of permutative equational theories \cite{Schmidt-Schauss89}.  Finally, exploring the relation between  Higher-Order Pattern unification in the style of Levy and Villaret~\cite{LevyVillaretRTA08} and nominal unification with fixed-point constrains would be also of great interest.

\bibliographystyle{alpha}
\bibliography{biblio}

%
%
%


\end{document}